\documentclass[a4paper,onecolumn,11pt]{article}
\usepackage{iftex}
\ifPDFTeX
\pdfoutput=1
\input{glyphtounicode}
\fi
\ifXeTeX\PassOptionsToPackage{xetex}{hyperref}\fi

\usepackage[T1]{fontenc}
\usepackage{lmodern}
\usepackage[margin=3cm]{geometry}
\usepackage{caption}
\usepackage{needspace}
\usepackage{etoolbox}
\usepackage[utf8]{inputenc}
\usepackage[provide=*,english]{babel}
\usepackage{microtype}
\usepackage{amsmath,amssymb,amsthm,mathtools,bm}
\usepackage{array}
\usepackage{booktabs}
\usepackage{graphicx}
\usepackage{enumitem}
\usepackage[numbers,sort&compress]{natbib}
\usepackage[colorlinks=true,linkcolor=blue,citecolor=blue,urlcolor=blue]{hyperref}
\usepackage[nameinlink,capitalize,noabbrev]{cleveref}

\hypersetup{
  pdftitle={Finite-Shot Operating Windows for Virtual Distillation and Symmetry Verification},
  pdfauthor={Vicenzo Scavino Alfaro},
  pdfsubject={Research manuscript},
  pdfkeywords={quantum error mitigation, virtual distillation, symmetry verification, finite-shot analysis}
}
\AtBeginEnvironment{thebibliography}{\raggedright}
\graphicspath{{figures/}}

\allowdisplaybreaks

\newtheorem{assumption}{Assumption}
\newtheorem{definition}{Definition}
\newtheorem{theorem}{Theorem}
\newtheorem{proposition}{Proposition}
\newtheorem{corollary}{Corollary}
\newtheorem{lemma}{Lemma}
\newtheorem{remark}{Remark}

\crefname{assumption}{Assumption}{Assumptions}
\Crefname{assumption}{Assumption}{Assumptions}

\DeclareMathOperator{\Tr}{Tr}

\DeclareMathOperator{\KL}{KL}
\newcommand{\MSE}{\operatorname{MSE}}
\newcommand{\Bias}{\operatorname{Bias}}
\newcommand{\Var}{\operatorname{Var}}
\newcommand{\Cov}{\operatorname{Cov}}
\newcommand{\E}{\mathbb{E}}
\newcommand{\Prob}{\mathbb{P}}

\newcommand{\1}{\mathbf{1}}
\newcommand{\OO}{\mathcal{O}}
\newcommand{\norm}[1]{\left\lVert #1 \right\rVert}
\newcommand{\abs}[1]{\left\lvert #1 \right\rvert}
\newcommand{\ip}[2]{\left\langle #1\middle|#2\right\rangle}
\newcommand{\bra}[1]{\left\langle #1\right|}
\newcommand{\ket}[1]{\left|#1\right\rangle}

\title{Finite-Shot Operating Windows for Virtual Distillation and Symmetry Verification}
\author{Vicenzo Scavino Alfaro\thanks{Email: \href{mailto:u201919346@upc.edu.pe}{u201919346@upc.edu.pe}. ORCID: \href{https://orcid.org/0009-0000-2472-9785}{0009-0000-2472-9785}.}\\[0.3em]
\normalsize Independent researcher}
\date{October 20, 2026}

\begin{document}
\maketitle

\begin{abstract}
Quantum error mitigation methods are usually compared by their infinite-shot bias, yet in an experiment the comparison is decided by a finite shot budget, by the instability of the estimator, and by the cost of each shot.  We derive finite-shot mean-squared-error laws for virtual distillation (VD) and symmetry verification (SV) whose leading terms predict when one method overtakes the other and whose remainders are bounded with explicit constants.  For VD the law contains the statistical bias and denominator instability of the quotient estimator, together with the sample size needed before the quotient can be trusted; for SV it separates the floor left by undetectable errors from the cost of postselection.  A first-order decomposition of the noise generator shows which error components each method removes.  Two such laws give a tie, a uniform dominance, or a trade-off whose crossing budget scales as \(p^{-2}\) or \(p^{-1}\) in the noise rate \(p\); when the remainder bounds are sharp enough, they turn the leading crossing into a certified interval.  In exact white-box QAOA simulations the predicted \(p^{-2}\) scaling (fitted exponent \(-1.97\)) and the sign pattern of all pairwise comparisons (\(300/300\) cells) are recovered, one pre-specified all-instance criterion is not met, and the crossings lie \(3\)--\(13\%\) above the leading prediction; a certified interval containing the exact crossing is obtained only in a closed model.  Gate-level simulations and runs on two IBM backends show why the ideal VD advantage disappears in the tested instances: interferometric overhead and denominator instability move its window beyond the tested budgets, where SV has the lower mean-squared error.
\end{abstract}

\section{Introduction}

Quantum error mitigation (QEM) aims to improve observable estimates on noisy devices without implementing full fault tolerance \citep{CaiReview2023}.  Standard techniques include zero-noise extrapolation, probabilistic error cancellation, symmetry verification, Clifford-data regression, and multi-copy purification methods such as virtual distillation and error suppression by derangements \citep{Temme2017,LiBenjamin2017,Endo2018,BonetMonroig2018,McArdle2019,Czarnik2021,Lowe2021,Huggins2021,Koczor2021PRX,HuoLi2022Dual,Hakoshima2024Localized,Yuan2024ResourceDistillation,Bultrini2023,Ijaz2024}.  A common issue is that lower bias is not free.  Mitigation can increase shot noise, circuit count, circuit width, calibration cost, or all of these.

We focus on two methods with analyzable finite-shot structure.  Virtual distillation (VD) estimates an observable with respect to
\begin{equation}
    \frac{\rho_p^M}{\Tr(\rho_p^M)},
\end{equation}
where \(M\) noisy copies are combined \citep{Huggins2021,Koczor2021PRX}.  VD suppresses subdominant spectral components, but it cannot remove the mismatch between the dominant eigenvector of the noisy state and the ideal state \citep{Koczor2021NJP}.  Symmetry verification (SV) projects samples onto a valid symmetry sector and discards invalid samples \citep{BonetMonroig2018,McArdle2019,Sagastizabal2019,Kakkar2022}.  Its residual bias comes from symmetry-preserving errors that postselection cannot detect; its variance cost comes from the acceptance probability.  For a symmetry-compatible observable the two first-order residual biases come from one decomposition of the noise generator: a coherent tilt of the ideal state that neither method removes, undetectable in-sector error that only VD removes, and detectable error that both remove.

The main mathematical issue is that VD is a quotient estimator.  The population target is
\begin{equation}
    \mu_{\mathrm{VD},M}(p)=\frac{N_M(p)}{D_M(p)}
    =\frac{\Tr(O\rho_p^M)}{\Tr(\rho_p^M)},
\end{equation}
while the finite-shot estimator is \(\widehat N_M/\widehat D_M\).  Therefore a VD MSE law must account not only for population bias and variance, but also for the statistical bias of a random quotient and for the probability that the estimated denominator approaches zero.  Prior VD analyses identify the population mechanism and noise floors \citep{Huggins2021,Koczor2021PRX,Koczor2021NJP}, while circuit-level and hardware purification work studies implementation noise and special cases \citep{Teo2022,Vikstal2024,OBrien2023Purification}.  We place these effects inside a finite-budget MSE law with explicitly bounded remainders.  Operationally, the denominator certificate is a physical resolution condition: the copy-overlap signal \(D_M\) must be large enough, relative to the available shot budget and circuit cost, to be distinguished from sampling fluctuations before a VD advantage can be claimed.

Related finite-budget and threshold viewpoints are complementary.  Cai's symmetry expansion framework compares symmetry-based schemes, including VD-type expansions, by bias and sampling cost, and already gives the leading crossover sample number beyond which a lower-bias, costlier estimator wins \citep{Cai2021}.  Our leading crossing formula has the same form; what is new is the control of its remainders, including the quotient bias and denominator instability of VD, which makes the crossing location certifiable, together with the decomposition of the first-order floors.  Shot-budget-aware regime identification has recently been explored for PEC in optimization benchmarking \citep{Demarty2026Benchmark}; the operating windows here follow from non-asymptotic MSE laws and include the quotient-estimator structure of VD.  Shot-estimation work for noisy circuits also separates statistical variance from a bias floor \citep{SeksariaPrabhakar2025}; here the same split underlies the method comparison.  Error-mitigation thresholds for imperfectly characterized random circuits address robustness to noise-model disorder \citep{Niroula2025Thresholds}, whereas our windows fix the transcript model and vary the sampling budget.

Our main finding is that neither method wins universally, and that leading-order window laws explain when and why the ideal VD advantage disappears.  In exact white-box experiments the predicted \(p^{-2}\) window scale is recovered with fitted exponent \(-1.97\), and the predicted sign pattern of all pairwise method comparisons holds in \(300/300\) cells; the crossings lie \(3\)--\(13\%\) above the leading prediction, and the remainder bounds certify a crossing only in closed toy models.  Under gate-level device simulation and on archived IBM hardware, the two mechanisms that the laws isolate, denominator instability and interferometric overhead, move the idealized VD window outside the tested resource range, and calibrated SV has the lower MSE in the tested QAOA instances.  Main results and assumptions are stated in \cref{sec:main_results_assumptions}, and limitations in \cref{sec:scope}.

\Needspace{8\baselineskip}
\section{Main results and assumptions}
\label{sec:main_results_assumptions}

The contribution has five parts.
\begin{enumerate}[leftmargin=2em]
    \item \textbf{VD quotient law.}  The VD mean-squared error has an explicit \(B^{-1}\) quotient bias and a denominator-driven variance.
    \item \textbf{Denominator scale.}  A concentration bound gives the sample size \(B_{\mathrm{den}}\) beyond which the VD quotient is trustworthy.
    \item \textbf{SV local law.}  The SV quotient has no \(B^{-1}\) quotient bias, and its sampling coefficient is \(\sigma_{\rm acc}^2/a\).
    \item \textbf{First-order floors.}  The VD floor \(\delta_{\mathrm{VD}}\) is the observable-weighted tilt of the ideal state under the noise generator, and the unmitigated, VD, and SV slopes obey \(\beta_0=\delta_{\mathrm{VD}}+D_{\rm in}+D_{\rm out}\) and \(\delta_{\mathrm{SV}}=\delta_{\mathrm{VD}}+D_{\rm in}\), which shows which error component each method removes.
    \item \textbf{Comparison rule and window location.}  Two methods are either tied, uniformly ordered, or in a trade-off whose crossing budget is located within an explicit interval once the remainders are bounded; near \(p=0\) the crossing sits at \(R_0=\Theta(p^{-2})\) when the leading one-shot costs of the two methods differ and at \(R_0=\Theta(p^{-1})\) when they match.
\end{enumerate}

\Cref{sec:experimental_validation} tests these predictions in exact white-box simulation, gate-level device simulation, and IBM hardware.

The assumptions are local and method-specific.  The noisy state is assumed to admit the first-order expansion in \cref{eq:noise_expansion}; the VD dominant-eigenvector formula assumes a pure ideal output and a simple leading eigenvalue of \(\rho_p\) on the considered noise range, which is automatic near \(p=0\) for a pure ideal state; quotient certificates assume bounded one-shot numerator-denominator variables and an explicit denominator-concentration threshold; SV statements assume a fixed symmetry projector.

\section{Problem setup and local laws}

Let \(\rho_\star\) be the ideal output state and let \(O=O^\dagger\) be a bounded observable.  The target is
\begin{equation}
    \mu_\star=\Tr(O\rho_\star).
\end{equation}
The noisy output state at physical noise scale \(p\geq0\) is \(\rho_p\).  We assume the local expansion
\begin{equation}
    \rho_p=\rho_\star+p\Delta+\OO(p^2),
    \qquad
    \Delta=\mathcal L(\rho_\star),
    \label{eq:noise_expansion}
\end{equation}
where \(\Delta=\Delta^\dagger\) and \(\Tr\Delta=0\).  The \(\OO(p^2)\) term is in any fixed matrix norm; all Hilbert spaces considered here are finite-dimensional.  The trace condition follows from \(\Tr\rho_p=1\).

A mitigation method \(m\) produces an estimator \(\widehat\mu_m\).  Its performance is
\begin{equation}
    \MSE_m(p,B)=\E[(\widehat\mu_m-\mu_\star)^2]
    =\Bias(\widehat\mu_m)^2+\Var(\widehat\mu_m),
\end{equation}
where \(B\) denotes the number of independent statistical samples for that method.

\begin{definition}[Certified local finite-shot law]
\label{def:certified_law}
A method \(m\) has a certified local finite-shot law at \(p\) if there exist \(b_m(p)\), \(v_m(p)\), \(c_m(p)\), constants \(C_m(p)\), and a threshold \(B_{\min,m}(p)\) such that, for every \(B\geq B_{\min,m}(p)\),
\begin{equation}
    \MSE_m(p,B)=b_m(p)^2+\frac{v_m(p)+2b_m(p)c_m(p)}{B}+\varrho_m(p,B),
    \qquad
    \abs{\varrho_m(p,B)}\leq \frac{C_m(p)}{B^2}.
    \label{eq:certified_law}
\end{equation}
\end{definition}

Here \emph{certified} means only that the remainder bound is rigorous and its constants are explicit; whether those constants are small enough to be useful is examined case by case below.  We write \(\widetilde v_m=v_m+2b_mc_m\) for the leading \(1/B\) coefficient.  The cross-term is the interaction between population bias and the statistical quotient bias \(c_m\), the ratio-estimator effect familiar from survey sampling \citep{Cochran1977}; \(c_m=0\) for a plain sample mean, and \(\widetilde v_m\) can be smaller than \(v_m\) when \(b_mc_m<0\).  The quotient estimators below first give a bound \(\abs{\varrho_m}\le C_m/B^2+C'_me^{-\gamma_mB}\); since \(B^2e^{-\gamma_mB}\) decreases for \(B\ge2/\gamma_m\), the exponential term is absorbed into \(C_m/B^2\) for every \(B\) above any threshold
\begin{equation}
    B_{\exp,m}\geq \frac{2}{\gamma_m}
    \quad\text{with}\quad
    C'_me^{-\gamma_mB_{\exp,m}}\le\frac{C_{m,\mathrm{abs}}}{B_{\exp,m}^2},
    \label{eq:exponential_absorption_threshold}
\end{equation}
at the cost of replacing \(C_m\) by \(C_m+C_{m,\mathrm{abs}}\).

For unmitigated estimation,
\begin{equation}
    \mu_0(p)=\Tr(O\rho_p),
    \qquad
    b_0(p)=\mu_0(p)-\mu_\star=\beta_0p+\OO(p^2),
    \qquad
    \beta_0=\Tr(O\Delta).
    \label{eq:unmitigated_bias}
\end{equation}
If the one-shot variance is \(\sigma_0^2(p)\), then
\begin{equation}
    \MSE_0(p,B)=b_0(p)^2+\frac{\sigma_0^2(p)}{B}.
    \label{eq:unmitigated_law_v1}
\end{equation}

The rest of the theory supplies the ingredients of \cref{eq:certified_law} for VD and SV and then compares them.  \Cref{sec:vd_population,sec:vd_quotient} give the VD bias floor, sampling coefficients, and validity scale \(B_{\mathrm{den}}\), and \cref{sec:copy_number} treats the copy number; \cref{sec:sv} gives the SV law and, in \cref{subsec:bias_decomposition}, the common decomposition of the two first-order floors; \cref{sec:crossings} compares two such laws under a common resource budget, which locates the operating windows.

\section{Virtual distillation: population law}
\label{sec:vd_population}

This section gives the population part of the VD law, its bias floor.  For \(M\geq2\), idealized VD targets
\begin{equation}
    \mu_{\mathrm{VD},M}(p)=\frac{N_M(p)}{D_M(p)},
    \qquad
    N_M(p)=\Tr(O\rho_p^M),
    \qquad
    D_M(p)=\Tr(\rho_p^M).
    \label{eq:vd_population_law}
\end{equation}
Assume the spectral decomposition
\begin{equation}
    \rho_p=\sum_{j\geq1}\lambda_j(p)\ket{\phi_j(p)}\bra{\phi_j(p)},
    \qquad
    \lambda_1(p)>\lambda_2(p)\geq\cdots\geq0.
\end{equation}
Then
\begin{equation}
    \mu_{\mathrm{VD},M}(p)=
    \frac{\sum_j\lambda_j(p)^M O_{jj}(p)}{\sum_j\lambda_j(p)^M},
    \qquad
    O_{jj}(p)=\bra{\phi_j(p)}O\ket{\phi_j(p)}.
\end{equation}
VD suppresses the subdominant eigenvalues, but it converges to the observable expectation in the dominant eigenvector:
\begin{equation}
    \lim_{M\to\infty}\mu_{\mathrm{VD},M}(p)=O_{11}(p).
\end{equation}
The infinite-copy floor is therefore
\begin{equation}
    b_{\mathrm{VD},\infty}(p)=O_{11}(p)-\mu_\star.
\end{equation}

With \(r_M(p)=\sum_{j>1}(\lambda_j/\lambda_1)^M\), dividing numerator and denominator by \(\lambda_1^M\) gives \(\mu_{\mathrm{VD},M}-O_{11}=\sum_{j>1}(\lambda_j/\lambda_1)^M(O_{jj}-O_{11})/(1+r_M)\), and \(\abs{O_{jj}-O_{11}}\le2\norm O_\infty\), so the finite-\(M\) leakage obeys
\begin{equation}
    \abs{\mu_{\mathrm{VD},M}(p)-O_{11}(p)}
    \leq2\norm{O}_\infty r_M(p).
    \label{eq:spectral_leakage_v2}
\end{equation}

\subsection{First-order endpoint coefficient in terms of the noise generator}

Assume now that \(\rho_\star=\ket{\psi_\star}\bra{\psi_\star}\) is pure.  Let
\begin{equation}
    Q_\star=I-\ket{\psi_\star}\bra{\psi_\star}.
\end{equation}
The eigenvalue \(1\) of \(\rho_\star\) is simple and separated from the zero eigenspace by a unit gap.  The \(\OO(p^2)\) operator term in \cref{eq:noise_expansion} changes the dominant eigenvector only at order \(\OO(p^2)\) by this unit gap.  Standard non-degenerate eigenvector perturbation \citep{Kato1995} gives
\begin{equation}
    \ket{\phi_1(p)}=\ket{\psi_\star}+p\ket{\psi^{(1)}}+\OO(p^2),
    \qquad
    \ket{\psi^{(1)}}=Q_\star\Delta\ket{\psi_\star},
    \qquad
    \ip{\psi_\star}{\psi^{(1)}}=0.
    \label{eq:eta_generator}
\end{equation}

\begin{proposition}[VD floor coefficient]
\label{prop:vd_delta_generator}
Under \cref{eq:noise_expansion} and \cref{eq:eta_generator}, for \(M=\infty\) and for every fixed \(M\ge2\),
\begin{equation}
    b_{\mathrm{VD},M}(p)=\delta_{\mathrm{VD}}p+\OO(p^2),
    \label{eq:finite_M_vd_delta_same_v1}
\end{equation}
with
\begin{equation}
    \delta_{\mathrm{VD}}
    =2\operatorname{Re}\bra{\psi_\star}OQ_\star\Delta\ket{\psi_\star}
    =2\operatorname{Re}\bra{\psi_\star}O(I-\ket{\psi_\star}\bra{\psi_\star})\mathcal L(\rho_\star)\ket{\psi_\star}.
    \label{eq:delta_vd_explicit}
\end{equation}
\end{proposition}

\begin{proof}
By expansion of the dominant eigenvector,
\begin{align}
    O_{11}(p)
    &=\bra{\psi_\star+p\psi^{(1)}}O\ket{\psi_\star+p\psi^{(1)}}+\OO(p^2)\\
    &=\mu_\star+2p\operatorname{Re}\bra{\psi_\star}O\ket{\psi^{(1)}}+\OO(p^2).
\end{align}
Substitute \(\ket{\psi^{(1)}}=Q_\star\Delta\ket{\psi_\star}\), which gives the case \(M=\infty\).  For finite \(M\ge2\), Weyl continuity \citep{HornJohnson2013} gives \(\lambda_1(p)=1+\OO(p)\) and \(\lambda_j(p)=\OO(p)\) for \(j>1\), so \(r_M(p)=\OO(p^M)\), and \cref{eq:spectral_leakage_v2} gives \(\mu_{\mathrm{VD},M}-O_{11}=\OO(p^2)\).
\end{proof}

\Cref{eq:delta_vd_explicit} is the local first-order form of the dominant-eigenvector mismatch of \citet{Koczor2021NJP}: the observable-weighted tilt of the ideal state into its orthogonal complement.

\section{Virtual distillation: denominator concentration and quotient MSE}
\label{sec:vd_quotient}

The sampling part of the VD law comes from the quotient estimator.  Let \((X_{M,s},Y_{M,s})_{s=1}^B\) be independent copies of a one-sample pair satisfying
\begin{equation}
    \E X_M=N_M(p),
    \qquad
    \E Y_M=D_M(p),
\end{equation}
and define
\begin{equation}
    \widehat N_M=\frac1B\sum_{s=1}^B X_{M,s},
    \qquad
    \widehat D_M=\frac1B\sum_{s=1}^B Y_{M,s},
    \qquad
    \widehat\mu_{\mathrm{VD},M}=\frac{\widehat N_M}{\widehat D_M}.
\end{equation}
Write \(\sigma_{N,M}^2(p)=\Var(X_M)\), \(\sigma_{D,M}^2(p)=\Var(Y_M)\), and
\(\sigma_{ND,M}(p)=\Cov(X_M,Y_M)\) for the one-shot variance and covariance
coefficients used in the quotient law below.
The ratio is meaningful only when the denominator is statistically separated from zero.

\begin{assumption}[Bounded one-shot VD estimator]
\label{ass:bounded_vd_pair_v4}
For fixed \((p,M)\), \(D_M(p)>0\), and the numerator-denominator one-shot pair satisfies
\begin{equation}
    \abs{X_M}\leq K_{N,M}(p),
    \qquad
    \abs{Y_M-D_M(p)}\leq K_{D,M}(p)
\end{equation}
almost surely.
\end{assumption}

\begin{lemma}[Bernstein denominator certificate]
\label{lem:denom_bernstein_v4}
Under \cref{ass:bounded_vd_pair_v4},
\begin{equation}
    \Prob\!\left(\abs{\widehat D_M-D_M}\ge\frac{D_M}{2}\right)
    \leq
    2
    \exp\!\left[-\frac{B D_M(p)^2}{8\sigma_{D,M}^2(p)+\frac{4}{3}K_{D,M}(p)D_M(p)}\right].
    \label{eq:bernstein_denominator_v4}
\end{equation}
Consequently, to make this bilateral bad-denominator probability at most \(\varepsilon\), it is sufficient that
\begin{equation}
    B\geq
    B_{\mathrm{den},M}(p,\varepsilon)
    :=
    \left(\frac{8\sigma_{D,M}^2(p)}{D_M(p)^2}
    +\frac{4K_{D,M}(p)}{3D_M(p)}\right)
    \log\!\frac2\varepsilon.
    \label{eq:denom_sample_threshold_v4}
\end{equation}
\end{lemma}

\begin{proof}
Apply Bernstein's inequality to both tails of \(B^{-1}\sum_s(Y_{M,s}-D_M)\) with threshold \(D_M/2\), and take a union bound \citep{Boucheron2013}.
\end{proof}

The ratio \(\sigma_{D,M}^2/D_M^2\) in \cref{eq:denom_sample_threshold_v4} is the denominator-driven inflation parameter.  If \(D_M\) is small, VD may still have a favorable population bias but require many samples before the ratio estimator is stable.

For a non-asymptotic statement define the clipped quotient
\begin{equation}
    \widehat\mu_{\mathrm{VD},M}^{\mathrm{clip}}
    =\frac{\widehat N_M}{\max\{\widehat D_M,D_M/2\}}.
    \label{eq:clipped_quotient_v4}
\end{equation}
It equals the ordinary quotient on the good event \(\widehat D_M\geq D_M/2\).  As written, the floor \(D_M/2\) involves the unknown population denominator, so \cref{eq:clipped_quotient_v4} is a proof device rather than an implementable estimator; it is directly usable only when \(D_M\) is known, as in the white-box Level~1 checks.  An implementable version replaces \(D_M/2\) by any known floor \(\tau\in(0,D_M/2]\), for instance a lower bound from \(D_M\ge\lambda_1(p)^M\) (\cref{eq:dm_spectral_bounds_v4}) when \(\lambda_1\) can be bounded, or from an independent pilot estimate.  The same proof then holds with \(D_M/2\) replaced by \(\tau\) in the bad-event sup-norm bounds, and the estimator still coincides with the ordinary quotient on the good event; the choice \(D_M/2\) gives the cleanest certificate.  Some floor is needed: for sign-valued denominator streams, \(\widehat D_M=0\) has positive probability at finite \(B\), so the ordinary quotient has no finite MSE; where the inverse moments of \(\widehat D_M\) are controlled, the ordinary quotient has the same leading expansion.  The same remarks apply to the SV floor \(a/2\) below.

The leading coefficients below are those of the classical delta-method quotient expansion \citep{VanDerVaart1998}; denominator concentration and clipping make the expansion non-asymptotic.

\begin{theorem}[VD quotient law]
\label{thm:vd_quotient_certified_v4}
Under \cref{ass:bounded_vd_pair_v4}, for every \(B\geq1\), the clipped VD estimator satisfies
\begin{align}
    \E[\widehat\mu_{\mathrm{VD},M}^{\mathrm{clip}}]
    &=\mu_{\mathrm{VD},M}(p)+\frac{c_{\mathrm{VD},M}(p)}{B}+r^{(E)}_{M}(p,B),
    \label{eq:vd_expectation_certified_v4}\\
    \Var(\widehat\mu_{\mathrm{VD},M}^{\mathrm{clip}})
    &=\frac{v_{\mathrm{VD},M}(p)}{B}+r^{(V)}_{M}(p,B),
    \label{eq:vd_variance_certified_v4}
\end{align}
where
\begin{align}
    c_{\mathrm{VD},M}(p)
    &=\frac{\mu_{\mathrm{VD},M}(p)\sigma_{D,M}^2(p)-\sigma_{ND,M}(p)}{D_M(p)^2},
    \label{eq:vd_c_v4}\\
    v_{\mathrm{VD},M}(p)
    &=\frac{\sigma_{N,M}^2(p)+\mu_{\mathrm{VD},M}(p)^2\sigma_{D,M}^2(p)
    -2\mu_{\mathrm{VD},M}(p)\sigma_{ND,M}(p)}{D_M(p)^2}
    \label{eq:vd_v_v4}\\
    &=\frac{\Var(X_M-\mu_{\mathrm{VD},M}(p)Y_M)}{D_M(p)^2}.
\end{align}
Consequently, with \(b_{\mathrm{VD},M}(p):=\mu_{\mathrm{VD},M}(p)-\mu_\star\),
\begin{equation}
    \MSE_{\mathrm{VD},M}(p,B)
    =b_{\mathrm{VD},M}(p)^2
    +\frac{v_{\mathrm{VD},M}(p)+2b_{\mathrm{VD},M}(p)c_{\mathrm{VD},M}(p)}{B}
    +\varrho_{\mathrm{VD},M}(p,B),
    \label{eq:vd_mse_certified_v4}
\end{equation}
\begin{equation}
    \abs{\varrho_{\mathrm{VD},M}(p,B)}
    \leq
    \frac{C_{\mathrm{VD},M}(p)}{B^2}
    +C'_{\mathrm{VD},M}(p)\,e^{-\gamma_{\mathrm{VD},M}(p)B},
    \qquad
    \gamma_{\mathrm{VD},M}=\frac{D_M^2}{8\sigma_{D,M}^2+\frac43 K_{D,M}D_M},
    \label{eq:vd_mse_remainder_v4}
\end{equation}
with explicit constants \(C_{\mathrm{VD},M}\), \(C'_{\mathrm{VD},M}\) given in \cref{app:constants}; the remainders \(r^{(E)}_M\), \(r^{(V)}_M\) obey bounds of the same form.  For \(B\ge B_{\exp,\mathrm{VD},M}\) as in \cref{eq:exponential_absorption_threshold}, the exponential term is absorbed into the \(B^{-2}\) term, so the VD estimator has a certified local law in the sense of \cref{def:certified_law}.
\end{theorem}

Proof in \cref{app:proofs}.

\begin{remark}[What the theorem is used for]
\label{rem:constant_tightness_v7}
The operational content of \cref{thm:vd_quotient_certified_v4} is the pair \((c_{\mathrm{VD},M},v_{\mathrm{VD},M})\) together with the denominator scale \(B_{\mathrm{den}}\) of \cref{lem:denom_bernstein_v4}.  The constants \(C_{\mathrm{VD},M}\), \(C'_{\mathrm{VD},M}\) come from worst-case moment envelopes; they show that the expansion is controlled for every \(B\), but they can exceed the true second-order correction by several orders of magnitude and are not estimates of the physical sampling penalty.  Model-specific bounds can be sharp: in the sector toy they certify the VD--SV crossing within about \(1\%\) (\cref{prop:sector_certified_crossing}).
\end{remark}

\subsection{Hadamard-test implementation}
\label{subsec:hadamard_v1}

To anchor the variance constants in a concrete implementation, we fix the following idealized estimator.  Let \(U_M\) be the cyclic derangement acting on \(M\) copies.  The denominator \(D_M=\Tr(\rho_p^M)\) is estimated by a controlled-\(U_M\) Hadamard test \citep{Ekert2002Direct}.  For a Pauli observable \(O\), the numerator \(N_M=\Tr(O\rho_p^M)\) is estimated by an independent controlled observable-derangement Hadamard test.  Indeed, the observable-derangement operator is a product of a Pauli unitary and \(U_M\), hence unitary, so the standard Hadamard test applies.  Each test returns a binary outcome in \(\{\pm1\}\).  We use \(B\) denominator shots and \(B\) numerator shots, and write \(\kappa_M\) for the circuit resource of one Hadamard test.

For this implementation,
\begin{equation}
    X_M\in\{\pm1\},
    \qquad
    Y_M\in\{\pm1\},
    \qquad
    \E X_M=N_M,
    \qquad
    \E Y_M=D_M,
\end{equation}
and the two streams are independent.  Therefore
\begin{equation}
    \sigma_{N,M}^2=1-N_M^2,
    \qquad
    \sigma_{D,M}^2=1-D_M^2,
    \qquad
    \sigma_{ND,M}=0.
    \label{eq:hadamard_variances_v4}
\end{equation}
The VD quotient constants become
\begin{align}
    c_{\mathrm{VD},M}^{\mathrm{Had}}(p)
    &=\frac{\mu_{\mathrm{VD},M}(p)\bigl(1-D_M(p)^2\bigr)}{D_M(p)^2},
    \label{eq:hadamard_c_v4}\\
    v_{\mathrm{VD},M}^{\mathrm{Had}}(p)
    &=\frac{1-N_M(p)^2+\mu_{\mathrm{VD},M}(p)^2\bigl(1-D_M(p)^2\bigr)}{D_M(p)^2}.
    \label{eq:hadamard_v_v4}
\end{align}
Since \(\norm{O}_\infty=1\) for a Pauli observable, positivity of \(\rho_p^M\) gives \(\abs{N_M}\le D_M\), hence \(\abs{\mu_{\mathrm{VD},M}}\le1\).  With \(K_{N,M}=1\) and \(K_{D,M}=2\) in \cref{ass:bounded_vd_pair_v4}, \cref{eq:denom_sample_threshold_v4} becomes explicit:
\begin{equation}
    B\geq
    \left(
    \frac{8(1-D_M^2)}{D_M^2}+\frac{8}{3D_M}
    \right)\log\!\frac2\varepsilon
    \label{eq:hadamard_Bmin_v4}
\end{equation}
is sufficient for \(\Prob(\abs{\widehat D_M-D_M}\ge D_M/2)\leq\varepsilon\).  The validation uses a spectrally normalized QAOA cost observable (\(\norm O_\infty=1\)), which is not unitary.  For any such \(O\), a numerator sample comes from the denominator circuit itself: reading the first copy in the eigenbasis of \(O\) together with the ancilla gives \(z\,o\), with \(z=\pm1\) the ancilla outcome and \(o\in[-1,1]\) the eigenvalue.  Its mean is \(\Tr(O\rho_p^M)=N_M\), and since the first copy is marginally distributed as \(\rho_p\), its variance \(\Tr(O^2\rho_p)-N_M^2\) does not exceed the \(\pm1\) value \(1-N_M^2\).  The Level~1 checks use the \(\pm1\) sign model, for which \cref{eq:hadamard_variances_v4,eq:hadamard_c_v4,eq:hadamard_v_v4} hold as written.  It is realizable from the same shots without further quantum calls: drawing \(S=\pm1\) with \(\Prob(S=+1\mid z\,o)=(1+z\,o)/2\) gives \(\E S=N_M\) and \(\Var S=1-N_M^2\).  At the same circuit cost the sign model therefore overstates, rather than understates, the VD numerator variance; the Level~2 interferometric circuit keeps \(z\,o\) directly.

\begin{remark}[One-sided convention for \(B_{\mathrm{den}}\)]
\label{rem:bden_onesided_convention}
The numerical diagnostics monitor the one-sided clip event \(\{\widehat D_M<D_M/2\}\), for which the Bernstein bound holds without the prefactor \(2\):
\begin{equation}
    \Prob\!\left(\widehat D_M<\frac{D_M}{2}\right)
    \leq
    \exp\!\left[-\frac{B D_M(p)^2}{8\sigma_{D,M}^2(p)+\frac{4}{3}K_{D,M}(p)D_M(p)}\right].
    \label{eq:bernstein_denominator_onesided}
\end{equation}
Every \(B_{\mathrm{den}}\) reported in \cref{sec:experimental_validation} is the corresponding one-sided threshold, \cref{eq:hadamard_Bmin_v4} with \(\log(2/\varepsilon)\) replaced by \(\log(1/\varepsilon)\) at \(\varepsilon=10^{-3}\); it is smaller than the two-sided threshold by the factor \(\log(1/\varepsilon)/\log(2/\varepsilon)\simeq0.91\).
\end{remark}

\begin{proposition}[Multiplicative implementation-noise shift of the VD law]
\label{prop:impl_noise_shift_v11}
Suppose the controlled-derangement/Hadamard tests are followed by effective sign contractions \(\widetilde N_M=\eta_{N,M}N_M\) and \(\widetilde D_M=\eta_{D,M}D_M\), with \(\eta_{N,M},\eta_{D,M}>0\) and \(\theta_M^{\rm impl}:=\eta_{N,M}/\eta_{D,M}\).  Then the implemented population quotient and bias are
\begin{equation}
    \mu_{\mathrm{VD},M}^{\rm impl}=\theta_M^{\rm impl}\frac{N_M}{D_M},
    \qquad
    b_{\mathrm{VD},M}^{\rm impl}
    =\left(\theta_M^{\rm impl}-1\right)\frac{N_M}{D_M}+b_{\mathrm{VD},M},
\end{equation}
and, for independent one-shot signs in \(\{\pm1\}\),
\begin{equation}
    v_{\mathrm{VD},M}^{\rm impl}
    =\frac{1-\eta_{N,M}^2N_M^2+\big(\theta_M^{\rm impl}N_M/D_M\big)^2(1-\eta_{D,M}^2D_M^2)}{\eta_{D,M}^2D_M^2}.
    \label{eq:impl_noise_vd_variance_v11}
\end{equation}
\end{proposition}

\begin{proof}
Take the ratio of the contracted means, and substitute them into \cref{eq:vd_v_v4} with independent numerator and denominator signs.
\end{proof}

This is an algebraic baseline for multiplicative readout or controlled-test contractions, narrower than a circuit-level noise analysis: compiled derangement circuits may create correlated, coherent, or non-multiplicative errors.  It is the falsifiable baseline tested in E10.  A common contraction (\(\theta_M^{\rm impl}=1\)) cancels in the population ratio, so the bias gap is unchanged, but it changes the sampling coefficient, which is not a pure rescaling of the ideal one; an asymmetric contraction adds a first-order bias.  Crossing budgets should therefore be computed from \cref{eq:impl_noise_vd_variance_v11}.

\subsection{Sharper variance diagnostics for the VD quotient}
\label{subsec:vd_sharper_variance_v11}

A bounded-differences argument gives a much tighter non-asymptotic bound on
the total variance of the clipped quotient.  It does not recover
\(v_{\mathrm{VD},M}\), but it avoids the moment envelopes behind
\(C_{\mathrm{VD},M}\).  It is used in E12.

\begin{lemma}[Efron--Stein total-variance certificate for clipped VD]
\label{lem:efron_stein_vd_v11}
Under \cref{ass:bounded_vd_pair_v4}, the clipped quotient in \cref{eq:clipped_quotient_v4} satisfies
\begin{equation}
    \Var\!\left(\widehat\mu_{\mathrm{VD},M}^{\mathrm{clip}}\right)
    \le
    \frac{8K_{N,M}^2\bigl(D_M+2K_{D,M}\bigr)^2}{D_M^4}\,\frac1B.
    \label{eq:efron_stein_vd_bound_v11}
\end{equation}
For Hadamard-test signs (\(K_{N,M}=1\), \(K_{D,M}=2\)) and \(D_M\in[1/2,1]\), the constant lies between \(2\times10^2\) and \(2.6\times10^3\).
\end{lemma}

\begin{proof}
Let \(f\) be the clipped quotient as a function of the \(B\) independent pairs \((X_s,Y_s)\).  Replace only the \(j\)-th pair and denote the resulting function by \(f^{(j)}\).  The sample numerator changes by at most \(2K_{N,M}/B\), while the sample denominator changes by at most \(2K_{D,M}/B\).  Since the clipped denominator is always at least \(D_M/2\), the map \((n,d)\mapsto n/\max\{d,D_M/2\}\) is Lipschitz with constants \(2/D_M\) in \(n\) and at most \(4K_{N,M}/D_M^2\) in \(d\), using \(\abs{\widehat N_M}\le K_{N,M}\).  Hence
\begin{equation}
    \abs{f-f^{(j)}}
    \le
    \frac{4K_{N,M}}{BD_M}
    +\frac{8K_{N,M}K_{D,M}}{BD_M^2}
    =
    \frac{4K_{N,M}(D_M+2K_{D,M})}{BD_M^2}.
\end{equation}
The Efron--Stein inequality \citep{Boucheron2013} gives
\begin{equation}
    \Var(f)\le \frac12\sum_{j=1}^B \E\big[(f-f^{(j)})^2\big]
    \le \frac{B}{2}\left(\frac{4K_{N,M}(D_M+2K_{D,M})}{BD_M^2}\right)^2,
\end{equation}
which is \cref{eq:efron_stein_vd_bound_v11}.
\end{proof}

\section{Copy-number choice for VD}
\label{sec:copy_number}

The copy number \(M\) is not a free improvement knob.  Increasing \(M\) suppresses subdominant spectral components, but it also lowers the denominator: with \(r_M\) as in \cref{eq:spectral_leakage_v2},
\begin{equation}
    D_M(p)=\lambda_1(p)^M\bigl(1+r_M(p)\bigr)\geq\lambda_1(p)^M,
    \label{eq:dm_spectral_bounds_v4}
\end{equation}
which increases the denominator-driven variance.

\begin{proposition}[VD implementation variance inflation]
\label{prop:implementation_variance_inflation_v4}
For the independent Hadamard-test implementation above,
\begin{equation}
    v_{\mathrm{VD},M}^{\mathrm{Had}}(p)
    \leq \frac{2}{D_M(p)^2}
    \leq 2\lambda_1(p)^{-2M}.
    \label{eq:implementation_upper_v4}
\end{equation}
If, in addition, \(1-N_M(p)^2\geq v_->0\) on the considered parameter range, then
\begin{equation}
    v_{\mathrm{VD},M}^{\mathrm{Had}}(p)\geq v_-\lambda_1(p)^{-2M}(1+r_M(p))^{-2}.
    \label{eq:implementation_lower_v4}
\end{equation}
\end{proposition}

\begin{proof}
The upper bound follows from \cref{eq:hadamard_v_v4} and \cref{eq:dm_spectral_bounds_v4}.  The lower bound uses the numerator-variance term \((1-N_M^2)/D_M^2\) and \(D_M=\lambda_1^M(1+r_M)\).
\end{proof}

\begin{proposition}[Logarithmic copy-number scale]
\label{thm:mstar_log_scale_v4}
Fix \(p\), let \(q=\lambda_2(p)/\lambda_1(p)<1\) with \(0<\lambda_2<\lambda_1\), and suppose that, on a range of \(M\) before the dominant-eigenvector floor takes over,
\begin{equation}
    A_-^2q^{2M}
    \leq
    \bigl(b_{\mathrm{VD},M}(p)-b_{\mathrm{VD},\infty}(p)\bigr)^2
    \leq
    A_+^2q^{2M},
    \label{eq:bias_two_sided_mstar_v4}
\end{equation}
and \(V_-\lambda_1^{-2M}\le v^{\mathrm{Had}}_{\mathrm{VD},M}\le V_+\lambda_1^{-2M}\) as in \cref{prop:implementation_variance_inflation_v4}.  Call \(M\) admissible if it lies in that range and \(\abs{b_{\mathrm{VD},\infty}(p)}\le A_-^2q^M/(4A_+)\).  There is a constant \(c_0\), independent of \(B\), such that whenever the integers within \(c_0+1\) of
\begin{equation}
    M_B=\frac{\log B}{2\abs{\log\lambda_2(p)}}
    \label{eq:mstar_log_lambda2_v4}
\end{equation}
are admissible, every admissible minimizer of the \(M\)-dependent part of the leading MSE, \(\mathcal E_M(B)=b_{\mathrm{VD},M}^2-b_{\mathrm{VD},\infty}^2+v^{\mathrm{Had}}_{\mathrm{VD},M}/B\), lies within \(c_0\) of \(M_B\).
\end{proposition}

The proof is in \cref{app:proofs}.  Since \(\lambda_1=1+\OO(p)\), \(M_B\approx\log B/(2\abs{\log(\lambda_2/\lambda_1)})\) at small noise; the exact denominator contains \(\abs{\log\lambda_2}\) because the variance inflation is driven by \(D_M\sim\lambda_1^M\) while leakage is driven by \((\lambda_2/\lambda_1)^M\).  Once \(\abs{b_{\mathrm{VD},\infty}}\) is comparable to the leakage term, \(M\) leaves the admissible range and the optimal copy number saturates.  At fixed physical budget \(R\) with \(\kappa_M\propto M\), the same argument gives \(M^\star=\log R/(2\abs{\log\lambda_2})-\OO\bigl(\log\log R/\abs{\log\lambda_2}\bigr)\).

\section{Symmetry verification}
\label{sec:sv}

SV is treated in the same two steps as VD, bias floor and then sampling law, after which \cref{subsec:bias_decomposition} compares the two floors.  Let \(P\) be the projector onto the valid symmetry sector and assume
\begin{equation}
    P\rho_\star P=\rho_\star.
\end{equation}
SV estimates
\begin{equation}
    \mu_{\mathrm{SV}}(p)=\frac{\Tr(OP\rho_pP)}{a(p)},
    \qquad
    a(p)=\Tr(P\rho_p);
    \label{eq:sv_population_v2}
\end{equation}
we call \(O\) symmetry-compatible when \([O,P]=0\).

\begin{proposition}[SV first-order bias]
\label{prop:sv_delta_precise}
Under \cref{eq:noise_expansion} and \(P\rho_\star P=\rho_\star\),
\begin{equation}
    a(p)=1-\alpha_{\mathrm{det}}p+\OO(p^2),
    \qquad
    b_{\mathrm{SV}}(p)=\delta_{\mathrm{SV}}p+\OO(p^2),
\end{equation}
where \(\alpha_{\mathrm{det}}=-\Tr(P\Delta)=\Tr(P^\perp\Delta)\), with \(P^\perp=I-P\), is the rate at which population leaves the sector, and
\begin{equation}
    \delta_{\mathrm{SV}}
    =\Tr(OP\Delta P)-\mu_\star\Tr(P\Delta P).
    \label{eq:delta_sv_precise}
\end{equation}
\end{proposition}

\begin{proof}
Since \(P\rho_pP=\rho_\star+pP\Delta P+\OO(p^2)\) and \(\Tr\Delta=0\), the projected state is \(P\rho_pP/a(p)=\rho_\star+p\bigl(P\Delta P-\rho_\star\Tr(P\Delta P)\bigr)+\OO(p^2)\); take the expectation of \(O\) and subtract \(\mu_\star\).
\end{proof}

\begin{proposition}[SV postselection quotient law]
\label{thm:sv_quotient_law_v2}
Let \(A\in\{0,1\}\) indicate acceptance on a shot and \(W\) be the value measured on accepted shots, and assume \(a(p)>0\) and \(\abs{W}\le K_W\); put \(Z=AW\), \(\mu_{\mathrm{SV}}=\E Z/a\), \(\sigma_{\mathrm{acc}}^2(p)=\Var(W\mid A=1)\), and \(\widehat\mu_{\mathrm{SV}}^{\rm clip}=\overline Z/\max\{\overline A,a/2\}\).  Then
\begin{equation}
    \Prob\left(\abs{\overline A-a}\ge a/2\right)
    \le 2e^{-Ba(p)/12},
    \label{eq:sv_acceptance_bernstein_v5}
\end{equation}
the quotient-bias coefficient vanishes exactly,
\begin{equation}
    c_{\mathrm{SV}}(p)
    =\frac{\mu_{\mathrm{SV}}(p)\Var(A)-\Cov(Z,A)}{a(p)^2}=0,
    \label{eq:sv_quotient_bias_zero_v14}
\end{equation}
and
\begin{equation}
    \MSE_{\mathrm{SV}}(p,B)
    =b_{\mathrm{SV}}(p)^2+
    \frac{v_{\mathrm{SV}}(p)}{B}
    +\varrho_{\mathrm{SV}}(p,B),
    \qquad
    v_{\mathrm{SV}}(p)=\frac{\sigma_{\mathrm{acc}}^2(p)}{a(p)},
    \label{eq:sv_certified_mse_v5}
\end{equation}
with \(\abs{\varrho_{\mathrm{SV}}(p,B)}\le C_{\mathrm{SV}}(p)B^{-2}+C'_{\mathrm{SV}}(p)e^{-Ba(p)/12}\) and explicit constants given in \cref{app:constants}.  In particular, if \(\sigma_{\mathrm{acc}}^2(p)=\sigma_\star^2+\sigma_1p+\OO(p^2)\) and \(a(p)=1-\alpha_{\mathrm{det}}p+\OO(p^2)\), then
\begin{equation}
    \MSE_{\mathrm{SV}}(p,B)
    =\delta_{\mathrm{SV}}^2p^2
    +\frac{\sigma_\star^2+(\sigma_1+\alpha_{\mathrm{det}}\sigma_\star^2)p}{B}
    +\OO(p^3)+\OO(p^2/B)+\varrho_{\mathrm{SV}}(p,B),
\end{equation}
with no separate \(\OO(p/B)\) quotient-bias cross term.
\end{proposition}

\begin{proof}
The two-sided multiplicative Chernoff bound \citep{Boucheron2013} gives \(\Prob(\abs{\overline A-a}\ge \delta a)\le 2e^{-Ba\delta^2/3}\) for \(0<\delta<1\); take \(\delta=1/2\).  Since \(Z=AW\), \(\Cov(Z,A)=\E Z-\E Z\,\E A=a\mu_{\mathrm{SV}}(1-a)=\mu_{\mathrm{SV}}\Var(A)\), so \(c_{\mathrm{SV}}=0\).  Since \(Z-\mu_{\mathrm{SV}}A=A(W-\mu_{\mathrm{SV}})\), \(\Var(Z-\mu_{\mathrm{SV}}A)=a\sigma_{\mathrm{acc}}^2\), which gives \(v_{\mathrm{SV}}\).  The remainder bound is \cref{thm:vd_quotient_certified_v4} under the substitution of \cref{app:constants}, with the Chernoff rate \(a/12\) in place of \(\gamma_{\mathrm{VD},M}\).  The local expansion follows from Taylor expansion of \(1/a(p)\) and \cref{prop:sv_delta_precise}.
\end{proof}

\subsection{Which first-order errors each method removes}
\label{subsec:bias_decomposition}

The two first-order floors \(\delta_{\mathrm{VD}}\) and \(\delta_{\mathrm{SV}}\) and the unmitigated slope \(\beta_0\) are three linear functionals of the same generator \(\Delta\).  Splitting \(\Delta\) along the ideal state, the rest of the valid sector, and the invalid sector relates them.

\begin{proposition}[First-order bias decomposition]
\label{prop:bias_decomposition_v1}
Let \(\rho_\star=\ket{\psi_\star}\bra{\psi_\star}\) be pure with \(P\ket{\psi_\star}=\ket{\psi_\star}\), and put \(\Pi=\ket{\psi_\star}\bra{\psi_\star}\), \(Q_P=P-\Pi\), \(P^\perp=I-P\),
\begin{equation}
    D_{\rm in}=\Tr\bigl[(O-\mu_\star)Q_P\Delta Q_P\bigr],
    \qquad
    D_{\rm out}=\Tr\bigl[(O-\mu_\star)P^\perp\Delta P^\perp\bigr].
\end{equation}
Then \(\delta_{\mathrm{SV}}=2\operatorname{Re}\bra{\psi_\star}OQ_P\Delta\ket{\psi_\star}+D_{\rm in}\), and, if \([O,P]=0\),
\begin{equation}
    \beta_0=\delta_{\mathrm{VD}}+D_{\rm in}+D_{\rm out},
    \qquad
    \delta_{\mathrm{SV}}=\delta_{\mathrm{VD}}+D_{\rm in}.
    \label{eq:bias_decomposition_v1}
\end{equation}
To first order, the coherent tilt \(\delta_{\mathrm{VD}}\) survives both methods, the undetectable in-sector error \(D_{\rm in}\) is removed by VD but not by SV, and the detectable error \(D_{\rm out}\) is removed by both.
\end{proposition}

\begin{proof}
Inserting \(I=\Pi+Q_\star\) on both sides of \(\Delta\) in \(\Tr(O\Delta)\) gives \(\mu_\star\bra{\psi_\star}\Delta\ket{\psi_\star}+2\operatorname{Re}\bra{\psi_\star}OQ_\star\Delta\ket{\psi_\star}+\Tr(OQ_\star\Delta Q_\star)\), and \(\Tr\Delta=0\) gives \(\bra{\psi_\star}\Delta\ket{\psi_\star}=-\Tr(Q_\star\Delta Q_\star)\); hence \(\beta_0=\delta_{\mathrm{VD}}+\Tr[(O-\mu_\star)Q_\star\Delta Q_\star]\).  The same computation with \(P=\Pi+Q_P\) in \cref{eq:delta_sv_precise} gives the formula for \(\delta_{\mathrm{SV}}\).  If \([O,P]=0\), then \(\bra{\psi_\star}OP^\perp=\bra{\psi_\star}P^\perp O=0\), so \(\bra{\psi_\star}OQ_\star=\bra{\psi_\star}OQ_P\) and the two tilt terms coincide; and since \(O-\mu_\star\) commutes with \(P\), the cross terms between \(Q_P\) and \(P^\perp\) vanish, so \(\Tr[(O-\mu_\star)Q_\star\Delta Q_\star]=D_{\rm in}+D_{\rm out}\).
\end{proof}

Consequently, for a symmetry-compatible observable the first-order bias gap between the two methods is
\begin{equation}
    \frac{b_{\mathrm{SV}}^2-b_{\mathrm{VD},M}^2}{p^2}
    =D_{\rm in}\bigl(2\delta_{\mathrm{VD}}+D_{\rm in}\bigr)+\OO(p),
    \label{eq:vd_sv_bias_gap_v1}
\end{equation}
so, at leading order, VD has the strictly smaller floor exactly when \(D_{\rm in}(2\delta_{\mathrm{VD}}+D_{\rm in})>0\), that is, when \(D_{\rm in}\ne0\) and either \(D_{\rm in}\) has the sign of \(\delta_{\mathrm{VD}}\) or \(\abs{D_{\rm in}}>2\abs{\delta_{\mathrm{VD}}}\).  The closed toys of \cref{app:toy_models} realize the three cases: a rank-one sector (\(Q_P=0\)), no tilt (\(\delta_{\mathrm{VD}}=0\)), and both.

\begin{remark}[Depolarizing noise carries no first-order SV bias in a fixed sector]
\label{rem:depol_sv_zero}
Let \(P\) project onto a fixed Hamming-weight sector, let \(O\) be \(Z\)-diagonal, and let \(P\rho_\star P=\rho_\star\).  For single-qubit depolarizing noise with rates \(c_i\), \(\Delta=\sum_i c_i\bigl[(I_i/2)\otimes\Tr_i\rho_\star-\rho_\star\bigr]\).  For a basis string \(z\) in the sector, flipping bit \(i\) leaves the sector, so the diagonal entry of \((I_i/2)\otimes\Tr_i\rho_\star\) at \(z\) is \(\tfrac12q(z)\), where \(q\) is the diagonal of \(\rho_\star\).  Since only diagonals enter traces against \(O\) and \(P\), \(\delta_{\mathrm{SV}}=\sum_i c_i\bigl[(\tfrac12\mu_\star-\mu_\star\tfrac12)-(\mu_\star-\mu_\star)\bigr]=0\).  In the language of \cref{prop:bias_decomposition_v1} this is \(D_{\rm in}=-\delta_{\mathrm{VD}}\): the in-sector incoherent error cancels the tilt, so SV is first-order unbiased while VD keeps a generic \(\delta_{\mathrm{VD}}\neq0\).  The cancellation needs the diagonality of \(O\), because \((I_i/2)\otimes\Tr_i\rho_\star\) matches \(\tfrac12\rho_\star\) only on the diagonal, and it needs the depolarizing structure: amplitude-damping, \(T_1/T_2\)-type, coherent, or other inhomogeneous sector-preserving errors generally give \(\delta_{\mathrm{SV}}\neq0\) through \cref{eq:delta_sv_precise}.
\end{remark}

\section{Operating-window crossings}
\label{sec:crossings}

With both methods in the form \cref{eq:certified_law}, they can now be compared.  Let \(R\) be a common physical resource budget.  If method \(m\) consumes \(\kappa_m(p)\) resource units per effective statistical sample, then \(B_m=R/\kappa_m(p)\).  A certified law \cref{eq:certified_law} with constant \(C_m^{(B)}(p)\) and threshold \(B_{\min,m}(p)\) becomes
\begin{equation}
    \MSE_m(p,R)=b_m(p)^2+\frac{s_m(p)}{R}+\varrho_m(p,R),
    \qquad
    \abs{\varrho_m(p,R)}\leq\frac{C_m(p)}{R^2}
    \label{eq:resource_law_v1}
\end{equation}
for \(R\geq R_{\min,m}(p)\), where
\begin{equation}
    s_m=\kappa_m\widetilde v_m,
    \qquad
    C_m=\kappa_m^2\,C_m^{(B)},
    \qquad
    R_{\min,m}=\kappa_m B_{\min,m}.
    \label{eq:resource_rescaling_v1}
\end{equation}
From here on \(C_m\) always denotes the resource-normalized constant; in the toy models below \(\kappa_m=1\) unless a physical call budget is stated explicitly.

\begin{remark}[Resource normalization]
\label{rem:resource_normalization_v13}
\(R\) is the physical resource being charged, so \(\kappa_m\) must include every stream, copy, and circuit call needed for one effective sample.  For the independent Hadamard-test implementation of \cref{subsec:hadamard_v1}, one paired VD sample uses a numerator and a denominator test, so \(\kappa_{\mathrm{VD},M}=2\kappa_M\) when \(R\) counts Hadamard-test calls.
\end{remark}

At a common budget the preferred method is the one with the lower MSE.  We call the range of \(R\) on which a method has the lowest MSE its \emph{operating window}; adjacent windows meet at a \emph{crossing}, a budget at which two MSE curves intersect.

\begin{proposition}[Comparison rule]
\label{thm:certified_crossing}
\label{thm:selection_trichotomy_v10}
Fix \(p\) and let two methods \(i,j\) obey
\begin{equation}
    \MSE_m(R)=b_m^2+\frac{s_m}{R}+\varrho_m(R),
    \qquad
    \abs{\varrho_m(R)}\leq\frac{C_m}{R^2},
    \qquad R\ge R_{\min},
\end{equation}
with continuous remainders.  Put \(g=b_j^2-b_i^2\), \(\Delta s=s_i-s_j\), and \(C=C_i+C_j\).
\begin{enumerate}[label=(\alph*),leftmargin=2.2em]
\item \emph{Dominance.}  If \(g>0\) and \(\Delta s\le0\), method \(i\) is strictly better for every \(R>\max\{R_{\min},\sqrt{C/g}\}\); if \(g=0\) and \(\Delta s<0\), for every \(R>\max\{R_{\min},C/\abs{\Delta s}\}\).  The same holds with \(i\) and \(j\) exchanged.
\item \emph{Trade-off.}  If \(g>0\) and \(\Delta s>0\), let \(R_0=\Delta s/g\) and
\begin{equation}
    \begin{gathered}
    \eta=\frac{4C}{gR_0^2},
    \qquad
    R_{\rm lo}=\frac{\eta}{2\bigl(1+\sqrt{1-\eta}\bigr)}\,R_0,\\
    R_L=\frac{1+\sqrt{1-\eta}}{2}\,R_0,
    \qquad
    R_U=\frac{1+\sqrt{1+\eta}}{2}\,R_0.
    \end{gathered}
    \label{eq:eta_crossing_condition}
\end{equation}
If \(\eta<1\) and \(R_L\geq R_{\min}\), there is a crossing
\begin{equation}
    R^\star\in[R_L,R_U],
    \qquad
    \abs{R^\star-R_0}\leq \frac{\eta R_0}{2}=\frac{2C}{gR_0};
    \label{eq:crossing_interval}
\end{equation}
method \(i\) (lower bias) is strictly better for every \(R>R_U\), and method \(j\) (lower sampling coefficient) for every \(R\in(\max\{R_{\min},R_{\rm lo}\},R_L)\).  For \(\eta<1/2\), \([R_L,R_U]\) lies inside \([R_0(1-\eta),R_0(1+\eta)]\), the band used in the pre-specified analysis of \cref{sec:experimental_validation}.  If moreover
\begin{equation}
    \sup_{R\in[R_L,R_U]}
    \abs{(\varrho_i-\varrho_j)'(R)}
    <\frac{\Delta s}{R_U^2},
    \label{eq:crossing_monotonicity_condition_v5}
\end{equation}
the crossing is unique in \([R_L,R_U]\).
\end{enumerate}
If \(g=\Delta s=0\), the leading laws coincide and nothing is certified at this order.
\end{proposition}

\begin{proof}
Write \(\Phi(R)=\MSE_i-\MSE_j=-g+\Delta s/R+\varrho(R)\) with \(\varrho=\varrho_i-\varrho_j\) and \(\abs{\varrho(R)}\le C/R^2\) for \(R\ge R_{\min}\).  (a) If \(g>0\) and \(\Delta s\le0\), then \(\Phi\le-g+C/R^2<0\) for \(R>\sqrt{C/g}\); if \(g=0\) and \(\Delta s<0\), then \(\Phi\le\Delta s/R+C/R^2<0\) for \(R>C/\abs{\Delta s}\).  (b) With \(x=R_0/R\) and \(C=\eta gR_0^2/4\),
\begin{equation}
    g\Bigl[x-1-\frac\eta4x^2\Bigr]\le\Phi(R)\le g\Bigl[x-1+\frac\eta4x^2\Bigr].
\end{equation}
For \(\eta<1\), the quadratic \(\frac\eta4x^2-x+1\) has the real roots \(x_\mp=2(1\mp\sqrt{1-\eta})/\eta\), with \(x_-=R_0/R_L\) and \(x_+=R_0/R_{\rm lo}\), and is negative strictly between them; hence \(\Phi>0\) on \((R_{\rm lo},R_L)\cap[R_{\min},\infty)\) and \(\Phi(R_L)\ge0\).  The quadratic \(\frac\eta4x^2+x-1\) has the unique positive root \(R_0/R_U\) and is negative below it; hence \(\Phi<0\) for \(R>R_U\) and \(\Phi(R_U)\le0\).  Continuity gives a zero in \([R_L,R_U]\).  The distance bound follows from \(1-\sqrt{1-\eta}\le\eta\) and \(\sqrt{1+\eta}-1\le\eta/2\); for \(\eta\le3/4\), \(\sqrt{1-\eta}\ge1-2\eta\) gives \(R_L\ge R_0(1-\eta)\).  Under \cref{eq:crossing_monotonicity_condition_v5}, \(\Phi'=-\Delta s/R^2+\varrho'<0\) on \([R_L,R_U]\), so the zero is unique.
\end{proof}

The leading indifference budget \(R_{i\leftrightarrow j}(p):=R_0=\Delta s/g\), which has the form of the crossover sample number of \citet{Cai2021}, is therefore accepted only if \(\eta<1\) and \(R_L\geq R_{\min}\).
\begin{corollary}[Local location of operating windows]
\label{cor:window_location_v10}
Suppose
\begin{equation}
    b_m(p)=\delta_m p+\OO(p^2),
    \qquad
    s_m(p)=\sigma_m^2+s_m^{(1)}p+\OO(p^2).
\end{equation}
Assume \(\delta_i^2\ne\delta_j^2\), so that \(b_j(p)^2-b_i(p)^2=\Theta(p^2)\).
If \(\sigma_i^2\ne\sigma_j^2\), a genuine bias--variance tradeoff has a leading window at scale \(R_0=\Theta(p^{-2})\).  If \(\sigma_i^2=\sigma_j^2\) but \(s_i^{(1)}\ne s_j^{(1)}\), the leading window is instead at scale \(R_0=\Theta(p^{-1})\).  The \(1/p\) crossings in the closed toy models below arise because the two compared laws have matched leading shot cost and differ first at order \(p\) in their sampling coefficients.
\end{corollary}

\begin{proof}
Substitute the displayed expansions into \(g=b_j^2-b_i^2\) and \(\Delta s=s_i-s_j\).  If the constant term of \(\Delta s\) is nonzero, then \(R_0=\Delta s/g=\Theta(p^{-2})\).  If that constant term cancels, then \(\Delta s=\Theta(p)\) while \(g=\Theta(p^2)\), giving \(R_0=\Theta(p^{-1})\).
\end{proof}

\paragraph{Other methods.}
The comparison rule needs only a certified local law for each method, so it is not restricted to VD and SV.  PEC fits directly.  If it cancels the noise through first order, \(b_{\mathrm{PEC}}(p)=\OO(p^2)\); since its estimator is a sample mean of quasiprobability-weighted outcomes,
\begin{equation}
    \MSE_{\mathrm{PEC}}(p,B)
    =b_{\mathrm{PEC}}(p)^2
    +\frac{\gamma_{\mathrm{PEC}}(p)^2\sigma_{\mathrm{PEC}}^2(p)}{B}
\end{equation}
exactly, where \(\gamma_{\mathrm{PEC}}(p)=1+\gamma_{\mathrm{PEC},1}p+\OO(p^2)\) is the quasiprobability one-norm overhead \citep{Temme2017,VanDenBerg2023PEC} and \(\gamma_{\mathrm{PEC}}^2\sigma_{\mathrm{PEC}}^2\) is the per-shot variance.  Thus \(c_{\mathrm{PEC}}=0\), \(C_{\mathrm{PEC}}=0\), and \(s_{\mathrm{PEC}}=\kappa_{\mathrm{PEC}}\gamma_{\mathrm{PEC}}^2\sigma_{\mathrm{PEC}}^2\).  CDR does not fit without further work: its bias and variance depend on the training distribution, regression model, target observable, and allocation of shots between training and target circuits \citep{Czarnik2021,Lowe2021}, so it stays outside the comparison rule unless a separate local law is proven for the chosen training model.

\paragraph{The three pairs compared in the experiments.}
With the resource convention of E8 and E9 (\(\kappa_{\mathrm{VD},2}=2\) for the two Hadamard streams, \(\kappa_{\mathrm{SV}}=\kappa_0=1\)), the inputs of \cref{thm:certified_crossing,cor:window_location_v10} are as follows.  For SV versus no mitigation, the accepted state at \(p=0\) is the ideal state, so \(s_{\mathrm{SV}}(0)=\sigma_\star^2=s_0(0)\): the leading variances match, any genuine window sits at \(R_0=\Theta(p^{-1})\), and whether it exists is decided by the order-\(p\) sampling coefficients of \cref{thm:sv_quotient_law_v2}; when they also favor SV, SV dominates uniformly.  For VD versus no mitigation, at \(p=0\) the Hadamard-sign variance is \(1-\mu_\star^2\) (since \(D_2(0)=1\)), which exceeds the direct-measurement variance \(\bra{\psi_\star}O^2\ket{\psi_\star}-\mu_\star^2\) unless \(O^2=I\) on the support, and the factor \(\kappa_{\mathrm{VD},2}=2\) separates them even for a Pauli observable; the leading variances therefore differ and the window sits at \(R_0=\Theta(p^{-2})\).  For VD versus SV, the bias gap is \cref{eq:vd_sv_bias_gap_v1} and the sampling gap again carries the factor \(\kappa_{\mathrm{VD},2}\), so the sign pattern of \(D_{\rm in}(2\delta_{\mathrm{VD}}+D_{\rm in})\) and \(\Delta s\) decides between dominance and a crossing.  These are the predictions tested in E8 and E9.

\section{Toy models with closed constants}
\label{sec:toy_models_overview}

Before the numerical tests, three closed toy models in \cref{app:toy_models} make the constants concrete.  In a one-qubit
toy, the SV quotient can be analysed exactly, and the
SV/no-mitigation crossing is certified at ordinary small-noise parameters
(\cref{prop:toy_certified_nonempty_v5}).  In a fixed-sector toy with
\(\delta_{\mathrm{SV}}\neq0\), direct remainder bounds certify the VD--SV crossing at \(p=10^{-2}\) within
about \(1\%\) of the leading prediction, and the exact crossing lies inside the certified interval
(\cref{prop:sector_certified_crossing}).  A third toy has both
\(\delta_{\mathrm{VD}}\) and \(\delta_{\mathrm{SV}}\) nonzero, the case probed by
the Level~1 experiments E3--E4.

\section{Experimental validation}
\label{sec:experimental_validation}

The numerical evidence has three levels (\cref{tab:validation_hierarchy_v1}).  Coefficients such as \(\delta_{\mathrm{VD}}\), \(\delta_{\mathrm{SV}}\), \(c_{\mathrm{VD}}\), and \(c_{\mathrm{SV}}\) are tested only in Level~1, where an exact white-box density-matrix model gives ground truth; Levels~2 and~3 test the form of the laws and the method ordering under device simulation and on hardware.  No experimental crossing is certified: \cref{thm:certified_crossing} needs remainder bounds, which are available only in the closed toys (\cref{prop:toy_certified_nonempty_v5,prop:sector_certified_crossing}), so the experimental crossing locations, computed with fitted or leading-law constants, are diagnostics.  The analysis plan was specified before the confirmatory runs and frozen by hash, in the spirit of preregistration \citep{Nosek2018}, though not filed with an external registry.  It lists predictions P1--P13 and experiments E1--E12, with P1--P6 forming the primary family; v1 and v2 denote the runs under the original plan and under its single amendment, which added a negative control to P1.  Both plan versions are archived with their hashes (\cref{sec:reproducibility}).

\begin{table}[!htbp]
\centering
\small
\begin{tabular}{@{}>{\raggedright\arraybackslash}p{0.16\linewidth}>{\raggedright\arraybackslash}p{0.30\linewidth}>{\raggedright\arraybackslash}p{0.44\linewidth}@{}}
\toprule
Level & Model & Claim supported \\
\midrule
Level~1 & Exact white-box density-matrix QAOA instances & Coefficient and structural checks against closed-form ground truth. \\
Level~2 & Seeded Qiskit/Aer gate-level simulation with synthetic thermal, depolarizing, and readout noise & Form and robustness of the operating-window laws under realistic state-preparation and VD-interferometry overhead. \\
Level~3 & Archived IBM \texttt{ibm\_marrakesh} hardware counts, with an \texttt{ibm\_kingston} cross-backend replicate, and calibrated readout & Qualitative feasibility, inter-instance, and cross-backend robustness on hardware; no coefficient-level inference. \\
\bottomrule
\end{tabular}
\caption{Experimental claim hierarchy.}
\label{tab:validation_hierarchy_v1}
\end{table}

\subsection{Level 1: white-box coefficient checks}

Level~1 uses ten fixed QAOA instances in a fixed Hamming-weight sector \citep{Farhi2014QAOA,Hadfield2019QAOA} and compares the fitted finite-shot laws with exact density-matrix quantities.  In this experimental subsection, \(\widehat{\beta}^{\rm fit}_0,\widehat{\beta}^{\rm fit}_1,\widehat{\beta}^{\rm fit}_2\) denote WLS coefficients in the empirical regression of MSE against \(1/B\); in particular, \(\widehat{\beta}^{\rm fit}_0\) estimates \(b^2\) and \(\widehat{\beta}^{\rm fit}_1\) estimates the leading MSE coefficient \(\widetilde v=v+2bc\).  These fit coefficients are distinct from the local unmitigated-bias coefficient \(\beta_0=\Tr(O\Delta)\) in \cref{eq:unmitigated_bias}.  \Cref{tab:level1_adjudication_v1} summarizes the primary checks.  In the v2 run, P1 did not meet its all-instance rule: the VD curvature signal was present in every instance, but the closed-form leading MSE coefficient \(\widetilde v\) fell outside the \(\widehat{\beta}^{\rm fit}_1\) interval in two of ten instances, \texttt{g6\_k2\_c000} and \texttt{g6\_k2\_c002}, with standardized deviations \(z=2.30\) and \(z=2.01\) (\cref{fig:coverage}); in the v1 run it was covered in all ten.

A post-run calibration puts this outcome in context.  With correct coefficients and nominal \(95\%\) intervals, a \(10/10\) conjunction succeeds with probability \(0.95^{10}\simeq0.60\), and \(8/10\) or fewer has probability about \(0.086\); since P1 combines several all-instance checks, even two independent ones would pass jointly with probability about \((0.95^{10})^2\simeq0.36\).  The pooled v2 residuals are consistent with nominal coverage (\(\chi^2\) \(p\simeq0.16\), mean \(z=-0.575\)).  An all-instance requirement therefore fails often even when the coefficients are correct.

The P9 crossing check tests the location of the window, not the certificate.  With \(C\) fitted to the sampled MSE, the empirical crossing lies above the leading \(R_0\) in all ten instances, by \(3\)--\(13\%\) (mean \(7.6\%\)), and inside the strict interval \([R_L,R_U]\) of \cref{thm:certified_crossing} in none, so the fitted constant underestimates the remainder.  The upward shift is consistent with a positive second-order VD remainder, which the exact computation in the sector toy also shows (\cref{prop:sector_certified_crossing}).

\begin{table}[!htbp]
\centering
\small
\begin{tabular}{@{}>{\raggedright\arraybackslash}p{0.18\linewidth}>{\raggedright\arraybackslash}p{0.25\linewidth}>{\raggedright\arraybackslash}p{0.47\linewidth}@{}}
\toprule
Check & Status & Main result \\
\midrule
P1/E1 VD quotient law & All-instance criterion not met & VD \(\Delta\mathrm{AIC}>2\) \citep{Akaike1974,BurnhamAnderson2002} in \(10/10\); VD curvature \(\widehat{\beta}^{\rm fit}_2\) excludes zero in \(10/10\); closed-form \(\widetilde v\) covered by \(\widehat{\beta}^{\rm fit}_1\) in \(8/10\) in v2 and \(10/10\) in v1. \\
P2/E1 intercept & Descriptive (P1 not met) & Closed-form \(b^2\) covered by \(\widehat{\beta}^{\rm fit}_0\) in \(9/10\). \\
P3--P4/E2 \(c_{\mathrm{VD}}\) vs \(c_{\mathrm{SV}}\) & Descriptive (P1 not met) & \(c_{\mathrm{VD}}\) passes \(30/30\) instance-\(p\) cells; \(c_{\mathrm{SV}}\) is within the fixed TOST margin \citep{Schuirmann1987} in \(29/30\). \\
P5--P6/E3--E4 first-order slopes & Descriptive (P1 not met) & Non-near-zero \(\delta_{\mathrm{VD}}\) cases pass \(7/7\); non-near-zero \(\delta_{\mathrm{SV}}\) cases pass \(9/9\). \\
P7/E5 denominator concentration & Secondary check met & Median bad-denominator probability drops from about \(0.091\) at \(0.125B_{\mathrm{den}}\) to \(9\times10^{-5}\) at \(B_{\mathrm{den}}\) and effectively zero at \(2B_{\mathrm{den}}\). \\
P9/E7 crossing band & Secondary check mixed & With \(C\) fitted to the sampled MSE, \(\eta<1/2\) in \(10/10\); the empirical crossing lies in the band in \(6/10\), in \([R_L,R_U]\) in \(0/10\), and above \(R_0\) in \(10/10\). \\
\bottomrule
\end{tabular}
\caption{Level~1 validation summary.  The primary family P1--P6 is tested in a fixed sequence, so once P1 is not met, P2--P6 are reported descriptively.}
\label{tab:level1_adjudication_v1}
\end{table}

\begin{table}[!htbp]
\centering
\footnotesize
\begin{tabular}{@{}>{\raggedright\arraybackslash}p{0.18\linewidth}>{\raggedright\arraybackslash}p{0.23\linewidth}>{\raggedright\arraybackslash}p{0.18\linewidth}>{\raggedright\arraybackslash}p{0.31\linewidth}@{}}
\toprule
Experiment & Paper result & Check & Secondary outcome \\
\midrule
E6/P8 & \cref{thm:mstar_log_scale_v4} & Copy-log scale mechanism & The pre-specified \(M^\star(B)\) regression is not resolvable because the floor is reached in only a few copies.  The mechanism checks pass: \(b_M-b_\infty\sim q^M\) has slope \(-5.23\) versus \(-5.04\), and \(\widetilde v_M\sim\lambda_1^{-2M}\) has slope \(0.1033\) versus \(0.1032\). \\
E8 & \cref{cor:window_location_v10} & Window location, leading law \(C=0\) & VD2-versus-unmitigated recovers the \(p^{-2}\) scale in \(10/10\) instances, with mean exponent \(-1.966\) and mean \(R^2\simeq0.9999\).  SV-versus-unmitigated has a valid positive crossing in \(0/10\) instances: the leading variances match at \(p=0\) and the order-\(p\) split also favors SV, giving uniform dominance. \\
E9 & \cref{thm:selection_trichotomy_v10} & Comparison rule, leading law \(C=0\) & Predicted and observed signs match in \(300/300\) cells.  VD-versus-unmitigated gives \(100/100\) crossings; SV-versus-unmitigated gives \(99/100\) dominance cells; VD-versus-SV gives \(73\) dominance cells and \(27\) crossings. \\
E10 & \cref{prop:impl_noise_shift_v11} & Implementation-noise baseline & Common-contraction bias invariance passes \(10/10\), the variance formula passes \(10/10\) at tolerance \(0.02\), and the asymmetric-bias shift passes \(10/10\) at \(4\) SE.  Mean variance inflation is \(1.24\) for \(\eta=0.9\) and \(1.57\) for \(\eta=0.8\). \\
E11 & \cref{prop:selector_lecam_constant_factor_v8} & Le Cam constant-scale arithmetic in the deriving toy & Arithmetic check of the remark, not independent QAOA evidence.  MSE divided by the conservative bound is \(\le64\) in \(72/72\), MSE divided by the local bound is \(\le32\) in \(72/72\), and \(sJ=1\) to \(10^{-12}\) in \(72/72\). \\
E12 & \cref{lem:efron_stein_vd_v11} & Efron--Stein operational certificate & The sampled certificate \(\Var_{\mathrm{emp}}\le ES/B\) holds in \(60/60\) cells.  The median ES constant is about \(1.03\times10^3\), compared with the formal \(10^8\) scale reference; under the ES law, SV dominates VD in \(5/10\) instances. \\
\bottomrule
\end{tabular}
\caption{Secondary Level~1 checks: descriptive white-box comparisons against closed-form quantities, outside the primary pre-specified family.}
\label{tab:level1_secondary_v1}
\end{table}

The secondary checks in \cref{tab:level1_secondary_v1} test the leading laws (\(C=0\)), not the worst-case remainder constants.

A companion E7 sensitivity analysis, also pre-specified, varies the resource and observable conventions.  Across the eight pre-specified
convention combinations---four resource definitions (paired sample, Hadamard
call, a two-qubit-gate proxy that charges each VD sample \(2[1+n(M-1)]\) against one for SV and unmitigated, and a wall-clock proxy) and two observable
normalizations---no pairwise comparison changes its large-resource winner in
any of the ten fixed instances: VD-versus-unmitigated retains a genuine
leading crossing in \(80/80\) cells, and the VD-versus-SV classification is
identical under every convention.  The only convention-dependent feature is
whether SV-versus-unmitigated appears as uniform SV dominance or as a finite
crossing: the wall-clock proxy, which charges SV a small postselection
bookkeeping cost, converts \(6\) of \(20\) instance--normalization cells from
dominance to a crossing without changing the large-resource winner.

The two remaining pre-specified robustness checks agree with the theory.  Replacing
the clipping floor \(D_M/2\) by \(D_M/3\) or \(2D_M/3\) leaves the fitted
leading coefficient unchanged on the fixed ensemble: with common random
numbers the three clipped estimators coincide on every draw without a clip
event, the clip rate is zero at the largest budgets, and the fitted
\(\widehat{\beta}^{\rm fit}_1\) values agree within \(2.4\%\) relative
deviation across all instances and thresholds.  Under a heterogeneous
depolarizing channel, the white-box first-order slopes follow the closed-form
generators in \(10/10\) instances, with maximum absolute error
\(5\times10^{-11}\) for \(\delta_{\mathrm{VD}}\); \(\delta_{\mathrm{SV}}\)
vanishes to machine precision, exactly as \cref{rem:depol_sv_zero} predicts,
while \(\delta_{\mathrm{VD}}\) remains generic.

\begin{figure}[!htbp]
\centering
\includegraphics[width=0.92\linewidth]{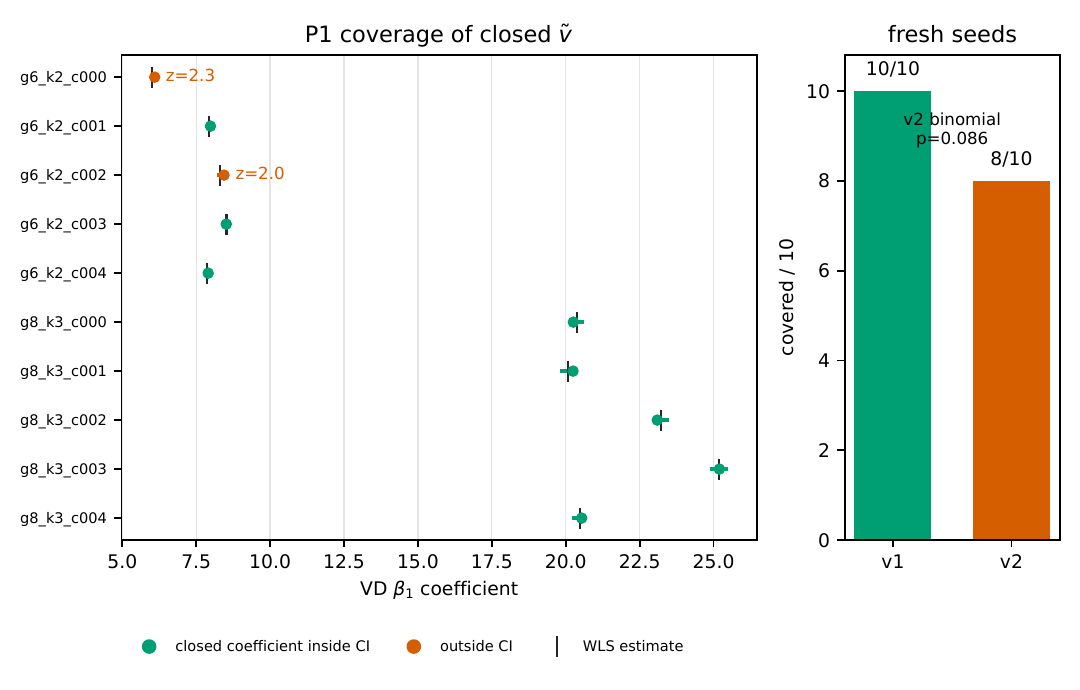}
\caption{Level~1 coverage check for the P1 leading VD MSE coefficient.  The horizontal intervals are WLS \(95\%\) intervals for \(\widehat{\beta}^{\rm fit}_1\), the dots are the closed-form \(\widetilde v\), and orange marks the two v2 instances outside the interval.}
\label{fig:coverage}
\end{figure}

\begin{figure}[!htbp]
\centering
\includegraphics[width=0.48\linewidth]{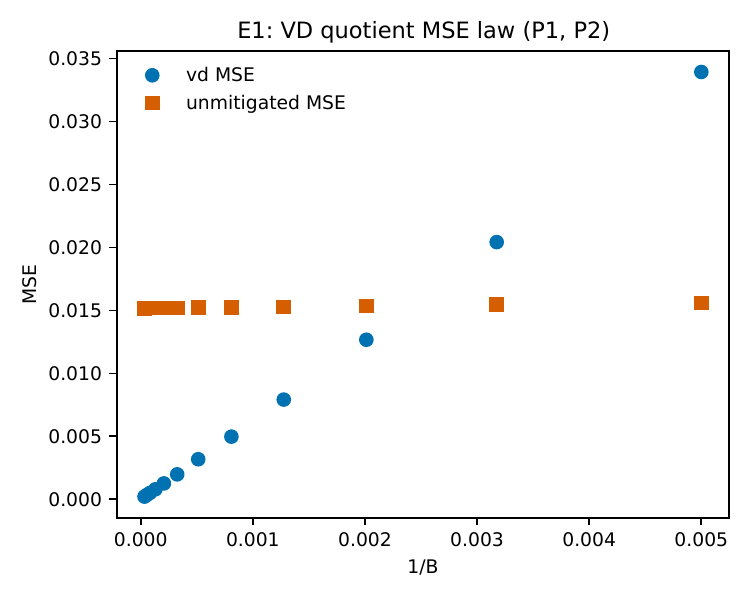}\hfill
\includegraphics[width=0.48\linewidth]{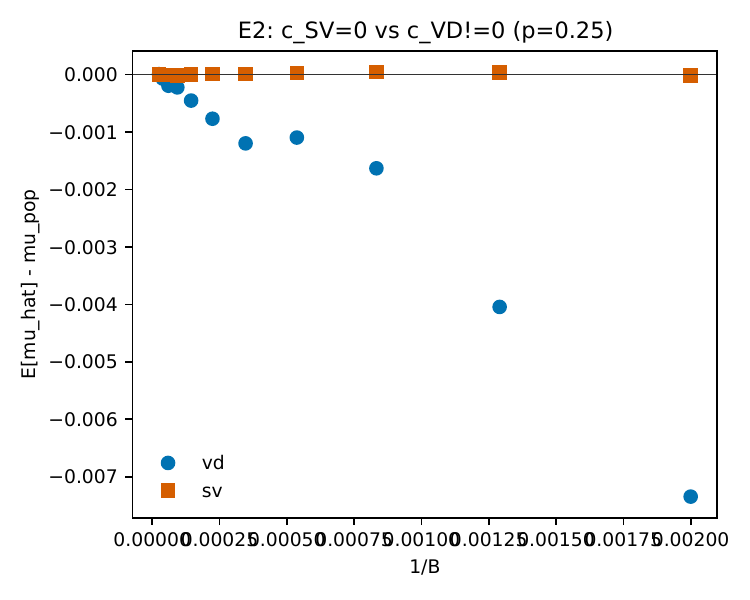}
\caption{Level~1 quotient-law and \(c\)-coefficient diagnostics from the v2 run.  Left: E1 MSE versus \(1/B\) for instance \texttt{g6\_k2\_c000}, showing the VD quotient curvature.  Right: E2 separates the nonzero VD quotient-bias slope from the SV slope centered near zero.}
\label{fig:coeff}
\end{figure}

\begin{figure}[t]
\centering
\includegraphics[width=0.48\linewidth]{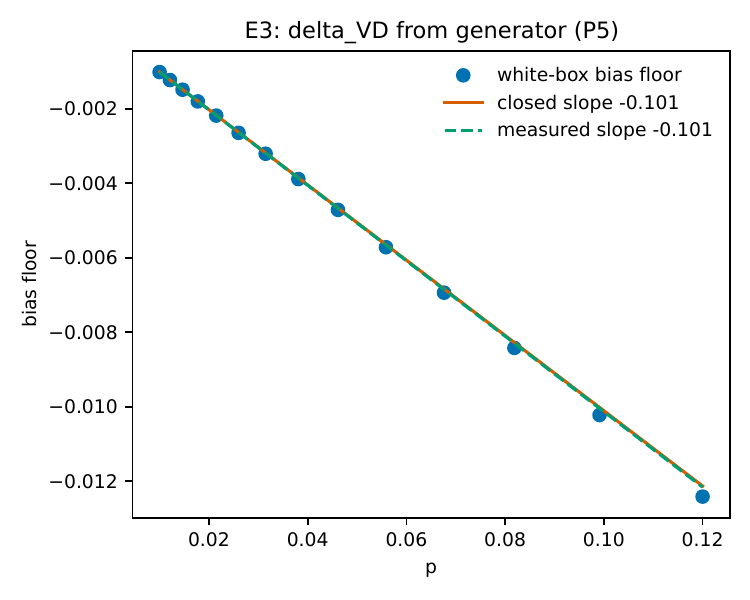}\hfill
\includegraphics[width=0.48\linewidth]{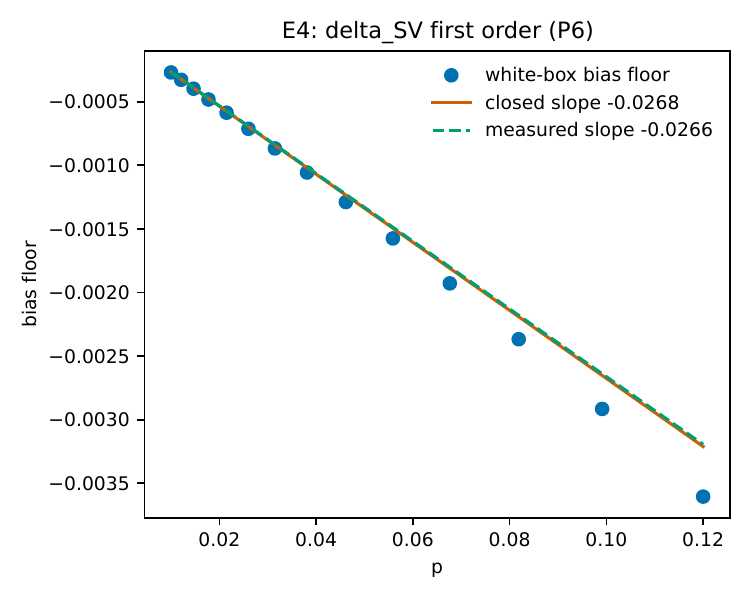}
\caption{Level~1 first-order bias-floor checks from the v2 run.  The measured white-box slopes track the generator formulas for \(\delta_{\mathrm{VD}}\) and \(\delta_{\mathrm{SV}}\) on instance \texttt{g6\_k2\_c000}.}
\label{fig:delta}
\end{figure}

\begin{figure}[t]
\centering
\includegraphics[width=0.80\linewidth]{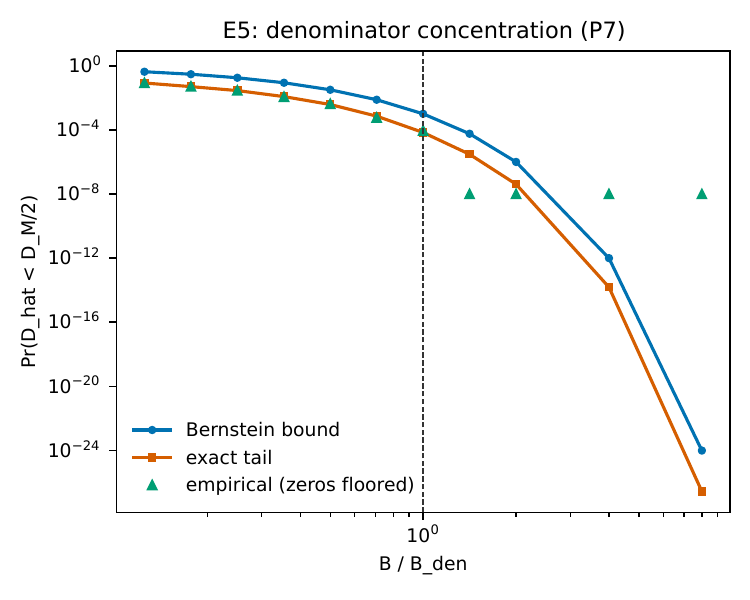}
\caption{Level~1 denominator-concentration diagnostic from the v2 run.  The bad-denominator probability collapses around the predicted scale \(B_{\mathrm{den}}\) (instance \texttt{g6\_k2\_c000}; all instances in the P7 row of \cref{tab:level1_adjudication_v1}).}
\label{fig:denominator}
\end{figure}

\begin{figure}[t]
\centering
\includegraphics[width=0.80\linewidth]{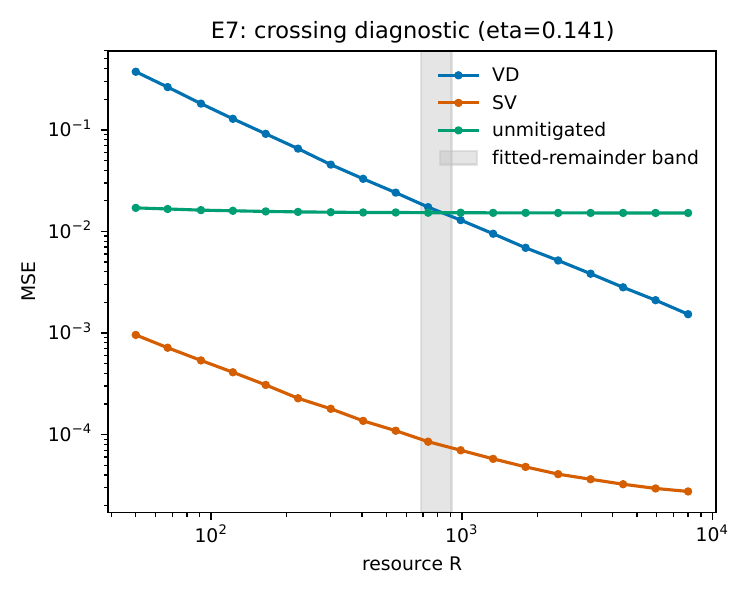}
\caption{Level~1 crossing diagnostic, VD versus unmitigated on instance \texttt{g6\_k2\_c000}, from the v2 run.  The shaded region is the P9 band \([R_0(1-\eta),R_0(1+\eta)]\), with \(\eta\) computed from a remainder constant \(C\) fitted to the sampled MSE curves rather than from the worst-case constants of \cref{app:constants}.}
\label{fig:crossing}
\end{figure}

\Cref{fig:coeff,fig:delta} display the coefficient-level diagnostics associated with P1--P6, while \cref{fig:denominator,fig:crossing} display the denominator and crossing diagnostics associated with P7 and P9.

\subsection{Level 2: device-simulation operating windows}

Level~2 asks whether the finite-shot operating-window laws survive in form under a gate-level Qiskit/Aer model \citep{Qiskit2024}.  The model uses Trotterized QAOA state preparation with a synthetic noise model, thermal relaxation and depolarizing errors on every gate plus readout errors, rather than a snapshot of a specific IBM backend; hardware is tested directly at Level~3.  The reference value is the noiseless Level~2 circuit, not the exact Level~1 dense-mixer state.  The pre-specified Level~2 validation metrics P10 and P13 are the device-level counterparts of the white-box checks E8 (window-location scaling) and E12 (Efron--Stein certificate), respectively.

In exact-density-matrix sampling runs, the finite-shot structure is visible: method-dependent population-bias floors plus sampling contributions that decrease with \(B\).  Pooled over the \(n=6\) and \(n=8\) instances, SV has the smallest median bias and MSE, while an optimistic state-noise-only \(M=2\) VD contrast roughly halves the unmitigated bias but triples its MSE (\cref{tab:level2_summary_v1}).  The \(n=6\) and \(n=8\) regimes differ: the \(n=8\) cases have much smaller \(D_2\simeq0.007\) and act as a denominator-concentration stress test.

Because this noise acts on every gate, it is not of the form of \cref{rem:depol_sv_zero}, which concerns depolarizing noise on the ideal output: errors during the circuit can be turned by later gates into in-sector population errors that postselection cannot detect, and the SV bias floor is indeed nonzero.  At Levels~2 and~3 the unmitigated estimator uses the raw counts, whereas SV and VD are readout-corrected, so the SV-versus-unmitigated gap includes the readout correction; the VD-versus-SV comparison is not affected.  The symmetry still favors SV in one respect: a single bit flip leaves the fixed-weight sector, so readout errors and relaxation in the final layer are detectable at first order.  The SV advantage found below is therefore a property of this symmetry and noise structure, not a general one.

The optimistic VD rows omit the cost of implementing the two-copy quotient on a device.  A realistic \(n=6\) follow-up includes Hadamard/SWAP interferometry, controlled-SWAP decomposition, gate noise, and readout.  In that circuit-level test, realistic VD has \(0/5\) bias-floor wins over SV, \(0/5\) large-\(B\) MSE wins at paired-sample resource, and \(0/5\) large-\(B\) MSE wins under the two-qubit-gate proxy.  The mechanism is visible in the overhead columns: for \texttt{g6\_k2\_c000}, realistic interferometry raises the VD bias floor by about \(2.8\) times relative to the optimistic state-noise-only contrast, and the median inflation over the five \(n=6\) instances is about \(3.3\) times.

The denominator and window diagnostics (\cref{tab:level2_summary_v1,fig:level2_window,fig:level2_denominator}) show the same mechanism.  Exact denominator-tail calculations show the bad-denominator probability collapsing around \(B_{\mathrm{den}}\), as at Level~1; the lowest grid point is \(0.1B_{\mathrm{den}}\) here and \(0.125B_{\mathrm{den}}\) at Level~1.  We call \emph{selected VD} the realistic interferometric circuit for \(n=6\) and, for \(n=8\), the optimistic contrast, which already loses to SV, so the realistic circuit would lose a fortiori.  The selected-VD-versus-SV comparison is SV-dominant or large-resource SV-winning in \(10/10\) instances under the two-qubit-gate proxy.  The optimistic \(n=6\) state-noise-only contrast, by contrast, has a VD large-resource window in \(2/5\) comparisons: the distillation mechanism produces a window, and the realistic interferometry overhead removes it.

\begin{table}[!htbp]
\centering
\small
\begin{tabular}{@{}>{\raggedright\arraybackslash}p{0.18\linewidth}>{\raggedright\arraybackslash}p{0.22\linewidth}>{\raggedright\arraybackslash}p{0.18\linewidth}>{\raggedright\arraybackslash}p{0.32\linewidth}@{}}
\toprule
Prediction & Level~2 result & State & Key evidence \\
\midrule
Finite-shot MSE shape & Bias floors plus finite-shot variance are visible & Validated in form & Pooled medians: absolute bias \(0.104\) (unmitigated), \(0.0145\) (SV), \(0.0549\) (optimistic VD); MSE \(0.0110\), \(3.65\times10^{-4}\), \(0.0332\). \\
P7/E5 denominator concentration & Bad-denominator probability collapses around \(B_{\mathrm{den}}\) & Validated in form & Median exact \(\Prob(\widehat D<D/2)\) is about \(0.12\) at \(0.1B_{\mathrm{den}}\), \(9.2\times10^{-5}\) at \(B_{\mathrm{den}}\), and \(1.8\times10^{-32}\) at \(10B_{\mathrm{den}}\). \\
P9/E7 crossings & SV dominates selected realistic/a-fortiori VD across the ensemble & Validated in form & SV is the large-resource winner versus unmitigated in \(10/10\); selected VD-versus-SV is SV-dominant or large-resource SV-winning in \(10/10\); the optimistic \(n=6\) VD window appears in \(2/5\). \\
P10/E8 folded sweep & Global unitary folding is outside the small-noise asymptotic regime & Limitation reported & Folded SV-versus-unmitigated has median slope \(0.063\) versus the \(-1\) asymptotic prediction; optimistic-VD-versus-SV has insufficient positive crossings. \\
P10/E8 small-noise variant (post hoc) & Direct small gate-noise scaling recovers P10 in form & Validated in form & Optimistic-VD-versus-unmitigated has median exponent \(-1.95\) with \(R^2\simeq0.9999\); SV-versus-unmitigated gives \(-0.99\) and \(-0.89\) only in the \(2/10\) instances with positive crossings. \\
P13/E12 Efron--Stein device bound & The ES threshold is not an operational VD-over-SV window & Validated in form as non-operational & Selected VD has finite two-qubit-gate-proxy \(R_{\mathrm{ES}}\) rows in \(0/10\) and dominates at \(R_{\mathrm{ES}}\) in \(0/10\); the optimistic contrast has maximum finite \(R_{\mathrm{ES}}\simeq5.2\times10^9\). \\
\bottomrule
\end{tabular}
\caption{Level~2 prediction-to-result map.  Optimistic VD is the state-noise-only contrast; realistic and selected VD are defined in the text.}
\label{tab:level2_summary_v1}
\end{table}

\begin{figure}[!htbp]
\centering
\includegraphics[width=0.82\linewidth]{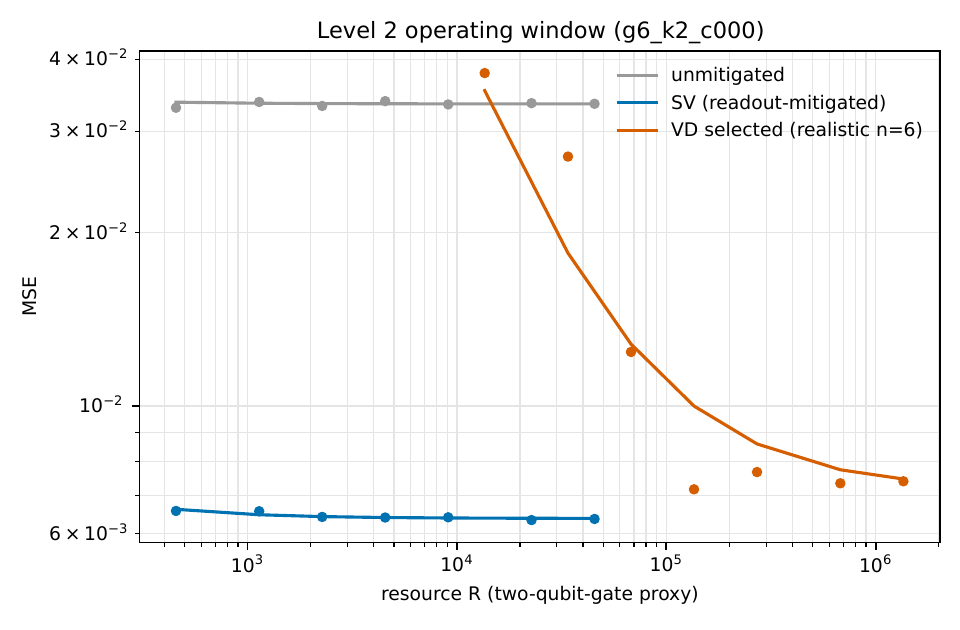}
\caption{Representative Level~2 operating-window curve for \texttt{g6\_k2\_c000} under the two-qubit-gate proxy.  Lines are fitted law curves and points are sampled MSE values from the E7/P9 Level~2 analysis.  The selected VD curve is the realistic \(n=6\) interferometry-including device circuit; SV remains below both unmitigated and selected VD on the plotted resource range.}
\label{fig:level2_window}
\end{figure}

\begin{figure}[!htbp]
\centering
\includegraphics[width=0.78\linewidth]{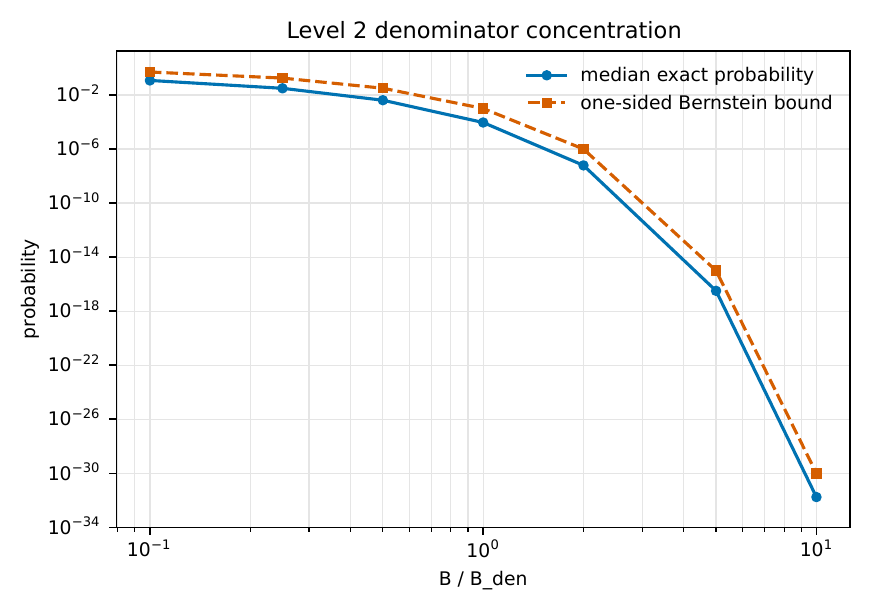}
\caption{Level~2 denominator concentration in the E5/P7 device-simulation layer.  The plot shows the median exact bad-denominator probability and the median one-sided Bernstein bound of \cref{eq:bernstein_denominator_onesided} versus \(B/B_{\mathrm{den}}\), computed from the same Level~2 denominator table used in \cref{tab:level2_summary_v1}.}
\label{fig:level2_denominator}
\end{figure}

The folded P10 sweep uses global unitary folding \citep{GiurgicaTiron2020} with \(\lambda\in\{1,3,5,7\}\), which amplifies depth rather than the small-\(p\) parameter of the theory; readout floors, base noise, and depth amplification then dominate, and the SV-versus-unmitigated median slope is about \(0.063\) instead of \(-1\).  Because global folding is outside the small-noise regime of the prediction, we added, after the folded sweep, a variant that switches readout error off and scales the gate error directly.  It recovers the predicted \(p^{-2}\) VD scaling and, in the \(2/10\) instances where SV-versus-unmitigated crossings can be fitted, the \(p^{-1}\) SV scaling (P10 small-noise row of \cref{tab:level2_summary_v1}).

In the Level~2 Efron--Stein (P13) analysis the constants of \cref{lem:efron_stein_vd_v11} are much larger than at Level~1, about \(1.6\times10^6\)--\(3.9\times10^6\) for realistic \(n=6\) selected VD and \(4.8\times10^{10}\)--\(6.5\times10^{10}\) for the \(n=8\) optimistic rows.  With them the certificate yields no practical VD-over-SV operating point in this regime (P13 row of \cref{tab:level2_summary_v1}).

\subsection{Level 3: hardware robustness}

Level~3 is an archive-only IBM-hardware robustness layer with a primary \texttt{ibm\_marrakesh} run and an independent cross-backend replicate on \texttt{ibm\_kingston}.  The analysis uses readout calibration for SV and VD \citep{Maciejewski2020Readout,Bravyi2021Readout} and cluster bootstrap over calibration epochs \citep{EfronTibshirani1993}, and asks whether the Level~2 ordering and denominator mechanism persist on hardware.

\begin{table}[!htbp]
\centering
\small
\begin{tabular}{@{}llcrrrl@{}}
\toprule
Backend & Instance & \(n\) & Epochs & SV MSE & Unmit.\ MSE & VD MSE \\
\midrule
\texttt{ibm\_marrakesh} & \texttt{g6\_k2\_c000} & 6 & 38 & \(0.0103\) & \(0.0450\) & \(0.0741\) \\
\texttt{ibm\_marrakesh} & \texttt{g6\_k2\_c001} & 6 & 16 & \(5.5\times10^{-4}\) & \(0.0151\) & \(0.0491\) \\
\texttt{ibm\_marrakesh} & \texttt{g8\_k3\_c000} & 8 & 12 & \(1.5\times10^{-4}\) & \(0.0219\) & diverges \\
\midrule
\texttt{ibm\_kingston}, fixed & \texttt{g6\_k2\_c000} & 6 & 14 & \(0.0103\) & \(0.0415\) & \(0.0489\) \\
\texttt{ibm\_kingston}, all 19 & \texttt{g6\_k2\_c000} & 6 & 19 & \(0.0103\) & \(0.0414\) & \(0.0503\) \\
\bottomrule
\end{tabular}
\caption{Largest-\(B\) hardware MSE at \(B=8192\): primary \texttt{ibm\_marrakesh} runs and the \texttt{ibm\_kingston} cross-backend replicate of \texttt{g6\_k2\_c000}.  Entries are point estimates rounded for readability; the ordering was assessed using epoch-cluster bootstrap.  The raw \(n=8\) VD point estimate, \(6.8\times10^6\), reflects denominator blow-up (\(\widehat D\approx0\)).  The \texttt{ibm\_kingston} replicate uses 14 epochs fixed by a result-blind stop rule; the all-19 row adds five later archived jobs as a sensitivity check.}
\label{tab:level3_hardware_v1}
\end{table}

The primary hardware ordering in \cref{tab:level3_hardware_v1} is consistent across two \(n=6\) instances and one \(n=8\) instance on \texttt{ibm\_marrakesh}.  In the \(n=8\) case, the \(M=2\) VD circuit requires a 17-qubit compiled circuit with depth \(3192\), and the readout-mitigated denominator is extremely small; the quotient blows up while SV remains stable.  In the hardware analysis the VD floor is a plug-in choice, half of a per-epoch reference denominator estimated from all archived denominator shots of that epoch (with a \(10^{-6}\) minimum), not a certified floor \(\tau\le D_M/2\), so it cannot prevent the \(n=8\) blow-up.  The independent \texttt{ibm\_kingston} replicate in \cref{tab:level3_hardware_v1} repeats the same \(B=8192\) ordering for \texttt{g6\_k2\_c000}, namely SV \(<\) unmitigated \(<\) VD, so the ordering is not specific to one backend calibration; adding the five later Kingston jobs leaves it unchanged.

The denominator mechanism of \cref{lem:denom_bernstein_v4} is visible on hardware.  For \texttt{g6\_k2\_c000}, the readout-mitigated VD denominator is small, with \(B/B_{\mathrm{den}}\simeq0.62\) at \(B=8192\) on \texttt{ibm\_marrakesh} and \(B/B_{\mathrm{den}}\simeq1.9\) on the fixed \texttt{ibm\_kingston} replicate.  The Marrakesh point remains below the denominator-concentration scale, while Kingston exceeds it at the largest \(B\) but VD still loses because of quotient bias and epoch-level variance.  No VD-over-SV crossing is resolved on the available hardware \(B\)-grid and budget.

\section{Scope and limitations}
\label{sec:scope}

The comparison is between concrete procedures and does not establish optimality over all QEM protocols; universal sampling-cost lower bounds address that question \citep{Takagi2022Limits,Takagi2023Sampling,Tsubouchi2023Cost,Quek2024Limits}, and a restricted two-point benchmark for a fixed transcript is recorded in \cref{app:toy_models}.  The main limitations are the following.
\begin{enumerate}[leftmargin=2em]
    \item The worst-case remainder constants are far from tight (\cref{rem:constant_tightness_v7}); sharp certificates come only from model-specific remainder bounds, available here for the closed toys (\cref{prop:toy_certified_nonempty_v5,prop:sector_certified_crossing}).
    \item The clipped estimator needs a known denominator floor; the hardware analysis uses a plug-in floor instead.
    \item The implementation-noise model is the idealized multiplicative baseline, whereas the validation uses compiled QAOA/VD/SV circuits.
    \item The hardware conclusion concerns the global interferometric VD circuits implemented here.  Localized virtual distillation, layout-aware compilation, or other lower-depth purification circuits change implementation constants such as \(\kappa_M\) and can move a crossing window, but the denominator, quotient-bias, and resource-normalization conditions still apply.
    \item The hardware layer is budget-limited (\cref{tab:level3_hardware_v1}), and only the \texttt{ibm\_kingston} replicate had its number of epochs fixed in advance; the epoch counts on \texttt{ibm\_marrakesh} were not fixed in advance.  It therefore supports the observed ordering rather than a hardware-universal one.
    \item All numerical evidence uses one problem family, QAOA in a fixed Hamming-weight sector, where single bit flips are detectable; other symmetries, or noise with more undetectable in-sector error, can shift the balance toward VD.
\end{enumerate}

\section{Conclusion}

Finite-shot error mitigation is a regime problem.  VD and SV can be compared by operating windows once their residual bias floors, sampling coefficients, quotient corrections, and remainder regimes are explicit.  The first-order floors follow from one decomposition of the noise generator: the coherent tilt survives both methods, the undetectable in-sector error is removed only by VD, and the detectable error by both.  On the sampling side, the key VD object is the quotient law together with the denominator scale \(B_{\mathrm{den}}\), and SV has no quotient bias.  Once each method has a law of the form \cref{eq:resource_law_v1}, method selection becomes a lower-envelope problem: the leading terms locate the windows, and the remainder bounds say when that location can be trusted.  The white-box checks agree with the coefficients and the predicted window scales and place the crossings \(3\)--\(13\%\) above the leading prediction; the crossing itself is certified only where the remainder can be bounded sharply, as in the closed toys, and remains a fitted quantity in the QAOA instances.  The device simulations and hardware runs show why an ideal VD advantage need not be an operational one: the two mechanisms isolated by the laws, implementation overhead and denominator instability, move the ideal VD window beyond the tested budgets, where calibrated SV has the lower MSE.

\section{Reproducibility and data availability}
\label{sec:reproducibility}

Level~1 is the reproducible core: a white-box density-matrix implementation with fixed instances and thresholds, and a script that checks its outputs against archived reference values; the same driver with the v2 random seeds regenerates the v2 run reported here.  Both versions of the analysis plan are archived with hashes, the frozen v2 files verify unchanged, and the Level~1 check runs without Qiskit or a quantum-cloud account.  Level~2 outputs come from seeded Qiskit/Aer runs and are archived with their seeds and reports.  Level~3 is reproducible at the analysis level: the raw IBM count archives, job identifiers, backend provenance, transpilation summaries, layouts, and calibration-epoch metadata are archived, and the analysis reruns from them without new QPU submissions.  Hardware credentials are not stored in the repository.

The exact release of the code and data behind the reported experiments (\texttt{v1.0.2}, a release number unrelated to the plan versions v1 and v2), with a hash-locked dependency file, is archived at Zenodo under DOI \href{https://doi.org/10.5281/zenodo.23560231}{10.5281/zenodo.23560231}, and the repository is public at \url{https://github.com/vicenzoscavino1999/finite-shot-vd-sv}.

\section*{Acknowledgements}

The author thanks his parents, Luigi Stefano Scavino Montero and Isabel Betzabe Alfaro Ninahuanca, and his sister, Luciana Franceska Scavino Alfaro, for their constant support, patience, and encouragement throughout this work.

\Needspace{10\baselineskip}
\section*{Author contribution statement}

Vicenzo Scavino Alfaro is the sole author of this work and is responsible for the conceptual development, theoretical analysis, software, numerical and hardware-data analysis, validation, visualization, writing, and final approval of the manuscript.

The author used AI-assisted tools for editorial and workflow support.  All scientific claims, calculations, code, data analysis, and manuscript text were reviewed and approved by the author, who remains fully responsible for the work.

\appendix
\numberwithin{equation}{section}

\section{Explicit quotient-remainder constants}
\label{app:constants}

This appendix assembles the explicit constants used in
\cref{thm:vd_quotient_certified_v4} and, through the substitution
\cref{eq:sv_substitution_v1}, in \cref{thm:sv_quotient_law_v2}.

\begin{definition}[Explicit quotient-remainder constants]
\label{def:quotient_constants_v4}
For fixed \((p,M)\), put \(\mu_M=N_M/D_M\), define the centered one-shot variables \(\xi=X_M-N_M\) and \(\zeta=Y_M-D_M\) and the signed moments \(m_{ab,M}=\E[\xi^a\zeta^b]\).  For sample means
\(\bar\xi=B^{-1}\sum_s \xi_s\) and \(\bar\zeta=B^{-1}\sum_s\zeta_s\), define finite moment envelopes
\begin{align}
    \mathfrak m_{12,M}(p)&=\sup_{B\ge1} B^2\abs{\E[\bar\xi\bar\zeta^2]},
    &
    \mathfrak m_{21,M}(p)&=\sup_{B\ge1} B^2\abs{\E[\bar\xi^2\bar\zeta]},\\
    \mathfrak m_{03,M}(p)&=\sup_{B\ge1} B^2\abs{\E[\bar\zeta^3]},
    &
    \mathfrak m_{13,M}(p)&=\sup_{B\ge1} B^2\E[\abs{\bar\xi}\abs{\bar\zeta}^3],\\
    \mathfrak m_{04,M}(p)&=\sup_{B\ge1} B^2\E[\abs{\bar\zeta}^4],
    &
    \mathfrak m_{40,M}(p)&=\sup_{B\ge1} B^2\E[\abs{\bar\xi}^4],\\
    \mathfrak m_{22,M}(p)&=\sup_{B\ge1} B^2\E[\bar\xi^2\bar\zeta^2].
\end{align}
These constants are finite under \cref{ass:bounded_vd_pair_v4}.  For centered sample means the third-order envelopes satisfy the exact identities
\[
    \mathfrak m_{12,M}=\abs{m_{12,M}},
    \qquad
    \mathfrak m_{21,M}=\abs{m_{21,M}},
    \qquad
    \mathfrak m_{03,M}=\abs{m_{03,M}},
\]
because only the fully diagonal index contractions survive.  A fully explicit admissible choice for the fourth-order envelopes follows from standard fourth-moment estimates for bounded centered averages; for example
\(\mathfrak m_{04,M}\le 3\sigma_{D,M}^4+K_{D,M}^4\),
\(\mathfrak m_{40,M}\le 3\sigma_{N,M}^4+16K_{N,M}^4\),
\(\mathfrak m_{22,M}\le\mathfrak m_{40,M}^{1/2}\mathfrak m_{04,M}^{1/2}\) by Cauchy--Schwarz, and \(\mathfrak m_{13,M}\le\mathfrak m_{40,M}^{1/4}\mathfrak m_{04,M}^{3/4}\) by H\"older's inequality.

For the variance certificate define the linear and quadratic quotient terms
\begin{equation}
    L_M=\frac{\bar\xi-\mu_M\bar\zeta}{D_M},
    \qquad
    Q_M=\frac{\mu_M\bar\zeta^2-\bar\xi\bar\zeta}{D_M^2},
\end{equation}
put \(G=\{\abs{\bar\zeta}\le D_M/2\}\), and let
\begin{equation}
    T_M=
    \begin{cases}
    \displaystyle
    \frac{\bar\xi\bar\zeta^2-\mu_M\bar\zeta^3}{D_M^3}
    -\frac{\bar\xi\bar\zeta^3}{D_M^4}
    +\frac{N_M+\bar\xi}{D_M}
      \frac{(\bar\zeta/D_M)^4}{1+\bar\zeta/D_M},
      &\abs{\bar\zeta}\le D_M/2,\\[1.0ex]
    0,&\abs{\bar\zeta}>D_M/2 .
    \end{cases}
    \label{eq:TM_good_event_definition_v1}
\end{equation}
Thus \(T_M\) is the signed sum of the third-and-higher Neumann terms after \(Q_M\), restricted to the good event where the displayed tail is bounded.  The remaining finite variance-tail envelope is
\begin{equation}
    \mathfrak t^{(V)}_M(p)
    :=\sup_{B\ge1}B^2
    \abs{2\Cov(L_M,T_M)+2\Cov(Q_M,T_M)+\Var(T_M)}.
    \label{eq:tail_variance_envelope_v14}
\end{equation}
It is finite under the boundedness assumption, and \cref{lem:bad_event_bookkeeping_v1}(v) in \cref{app:proofs} bounds it explicitly by sixth-moment envelopes of \((\bar\xi,\bar\zeta)\).
Define
\begin{align}
    A_{E,M}(p)
    &=\frac{\mathfrak m_{12,M}(p)+\abs{\mu_M}\mathfrak m_{03,M}(p)}{D_M(p)^3}
      +\frac{\mathfrak m_{13,M}(p)}{D_M(p)^4}
      +\frac{2K_{N,M}(p)\mathfrak m_{04,M}(p)}{D_M(p)^5},
       \label{eq:AE_constant_v4}\\
    A_{Q,M}(p)
    &=\frac{2}{D_M(p)^4}\left[
      \abs{\mu_M}^2\mathfrak m_{04,M}(p)+\mathfrak m_{22,M}(p)
      \right],
      \label{eq:AQ_constant_v14}\\
    A_{LQ,M}(p)
    &=\frac{2}{D_M(p)^3}\left[
      2\abs{\mu_M}\mathfrak m_{12,M}(p)+\mathfrak m_{21,M}(p)
      +\abs{\mu_M}^2\mathfrak m_{03,M}(p)
      \right],
      \label{eq:ALQ_constant_v14}\\
    A_{V,M}(p)
    &=A_{Q,M}(p)+A_{LQ,M}(p)+\mathfrak t^{(V)}_M(p),
      \label{eq:AV_constant_v4}\\
    \omega_M(p)&=\frac{K_{D,M}(p)}{D_M(p)},\\
    H_{E,M}(p)
    &=\frac{2K_{N,M}(p)}{D_M(p)}
      \left(5+3\omega_M(p)+3\omega_M(p)^2+3\omega_M(p)^3\right),\\
    H_{V,M}(p)&=2H_{E,M}(p)^2,\\
    E_{\mathrm{bad},M}(p,B)
    &=H_{V,M}(p)
    \exp\!\left[-\frac{B D_M(p)^2}{8\sigma_{D,M}^2(p)+\frac{4}{3}K_{D,M}(p)D_M(p)}\right].
\end{align}
\Cref{eq:ALQ_constant_v14} controls the mixed term \(2\Cov(L_M,Q_M)\), which bounding \(\Var(Q_M)\) alone would miss.
\end{definition}

With these constants, the remainders in \cref{thm:vd_quotient_certified_v4} obey, for every \(B\ge1\),
\begin{align}
    \abs{r^{(E)}_{M}(p,B)}
    &\leq \frac{A_{E,M}(p)}{B^2}+H_{E,M}(p)\,e^{-\gamma_{\mathrm{VD},M}(p)B},
    \label{eq:explicit_E_remainder_v4}\\
    \abs{r^{(V)}_{M}(p,B)}
    &\leq \frac{A_{V,M}(p)}{B^2}+E_{\mathrm{bad},M}(p,B),
    \label{eq:explicit_V_remainder_v4}
\end{align}
and \cref{eq:vd_mse_remainder_v4} holds with the admissible choice
\begin{align}
    C_{\mathrm{VD},M}
    &=A_{V,M}+c_{\mathrm{VD},M}^2
      +2\abs{b_{\mathrm{VD},M}}A_{E,M}
      +2\abs{c_{\mathrm{VD},M}}A_{E,M}
      +2A_{E,M}^2,\notag\\
    C'_{\mathrm{VD},M}
    &=H_{V,M}+2\abs{b_{\mathrm{VD},M}}H_{E,M}
      +2\abs{c_{\mathrm{VD},M}}H_{E,M}+2H_{E,M}^2.
    \label{eq:vd_mse_constants_v1}
\end{align}
The denominator scale \(B_{\mathrm{den},M}(p,\varepsilon)\) of \cref{lem:denom_bernstein_v4} is the high-probability scale that makes \(\Prob(G^c)\le\varepsilon\); it is not needed for these inequalities, which hold for every \(B\ge1\).

For SV, the replacement
\begin{equation}
    (N_M,D_M,X_M,Y_M,K_{N,M},K_{D,M})
    \mapsto
    (a\mu_{\mathrm{SV}},a,Z,A,K_W,1)
    \label{eq:sv_substitution_v1}
\end{equation}
in \cref{def:quotient_constants_v4} gives constants \(A_{E,\mathrm{SV}}\), \(A_{V,\mathrm{SV}}\), \(H_{E,\mathrm{SV}}\), and \(H_{V,\mathrm{SV}}=2H_{E,\mathrm{SV}}^2\), with \(\omega_{\mathrm{SV}}=1/a\).  With the Chernoff rate \(a/12\) of \cref{eq:sv_acceptance_bernstein_v5} in place of \(\gamma_{\mathrm{VD},M}\), the remainders of the SV expectation and variance expansions obey \(\abs{r^{(E)}_{\mathrm{SV}}}\le A_{E,\mathrm{SV}}/B^2+H_{E,\mathrm{SV}}e^{-Ba/12}\) and \(\abs{r^{(V)}_{\mathrm{SV}}}\le A_{V,\mathrm{SV}}/B^2+H_{V,\mathrm{SV}}e^{-Ba/12}\).  Since \(c_{\mathrm{SV}}=0\), one admissible choice in \cref{thm:sv_quotient_law_v2} is
\begin{equation}
    C_{\mathrm{SV}}=A_{V,\mathrm{SV}}+2\abs{b_{\mathrm{SV}}}A_{E,\mathrm{SV}}+2A_{E,\mathrm{SV}}^2,
    \qquad
    C'_{\mathrm{SV}}=H_{V,\mathrm{SV}}+2\abs{b_{\mathrm{SV}}}H_{E,\mathrm{SV}}+2H_{E,\mathrm{SV}}^2.
\end{equation}

\section{Deferred proofs}
\label{app:proofs}

This appendix contains an explicit bookkeeping lemma for the bad-event and tail terms of \cref{def:quotient_constants_v4} and the proofs of \cref{thm:vd_quotient_certified_v4} and \cref{thm:mstar_log_scale_v4}.

\setcounter{lemma}{0}
\renewcommand{\thelemma}{\thesection.\arabic{lemma}}
\renewcommand{\theHlemma}{proofsapp.\arabic{lemma}}

\begin{lemma}[Explicit bad-event and tail bookkeeping]
\label{lem:bad_event_bookkeeping_v1}
Under \cref{ass:bounded_vd_pair_v4}, fix \((p,M)\), drop the index \(M\), and put \(\omega_N=K_N/D\), \(\omega=K_D/D\), \(G=\{\abs{\bar\zeta}\le D/2\}\), \(\pi_B=\Prob(G^c)\), \(\gamma=D^2/(8\sigma_D^2+\frac43K_DD)\), and \(T^{(3)}=(\bar\xi\bar\zeta^2-\mu\bar\zeta^3)/D^3\).  Let \(L,Q,T\) be as in \cref{def:quotient_constants_v4}.
\begin{enumerate}[label=(\roman*),leftmargin=2.2em]
\item Everywhere, \(\abs{\bar\xi}\le2K_N\), \(\abs{\bar\zeta}\le K_D\), \(\abs{\mu}\le\omega_N\), and \(\abs{\widehat\mu^{\rm clip}}\le2\omega_N\).  Hence
\(\abs{L}\le\omega_N(2+\omega)\), \(\abs{Q}\le\omega_N\omega(2+\omega)\), \(\abs{T^{(3)}}\le\omega_N\omega^2(2+\omega)\), and \(\abs{\widehat\mu^{\rm clip}-\mu}\le3\omega_N\).
\item If \(\omega<1/2\), then \(G^c=\varnothing\), and every bad-event term below vanishes.
\item \emph{Expectation.}  With \(c=c_{\mathrm{VD},M}\),
\begin{equation}
    \E[\widehat\mu^{\rm clip}]-\mu-\frac cB
    =\E[T^{(3)}]
    +\E\!\left[\Big(-\frac{\bar\xi\bar\zeta^3}{D^4}+\widetilde{\mathcal R}_4\Big)\1_G\right]
    +\E\!\left[\big(\widehat\mu^{\rm clip}-\mu-L-Q-T^{(3)}\big)\1_{G^c}\right],
    \label{eq:bad_event_expectation_split_v1}
\end{equation}
and the last term is bounded by
\(\omega_N(5+3\omega+3\omega^2+\omega^3)\,\pi_B\le\frac12H_{E,M}\,\pi_B\le H_{E,M}e^{-\gamma B}\).
\item \emph{Variance.}  Let \(\widehat\mu_G=\mu+L+Q+T\) and \(\widehat\mu_{\rm bad}=\widehat\mu^{\rm clip}-\widehat\mu_G\), and put \(\Pi'=5+3\omega+\omega^2\).  Then
\begin{align}
    \abs{\Var(\widehat\mu^{\rm clip})-\Var(\widehat\mu_G)}
    &\le5\omega_N^2\Pi'^2\,\pi_B
    \le10\omega_N^2\Pi'^2e^{-\gamma B}\notag\\
    &\le2H_{E,M}^2e^{-\gamma B}=E_{\mathrm{bad},M}(p,B).
\end{align}
\item \emph{Tail envelope.}  Let \(\tau=\sup_{B\ge1}B^3\E[T^2]\), \(\mathfrak M_{06}=\sup_{B\ge1}B^3\E[\bar\zeta^6]\), and \(\mathfrak M_{60}=\sup_{B\ge1}B^3\E[\bar\xi^6]\).  Then
\begin{align}
    \mathfrak M_{06}&\le15\sigma_D^6+25K_D^2\sigma_D^4+K_D^4\sigma_D^2,
    \qquad
    \mathfrak M_{60}\le15\sigma_N^6+100K_N^2\sigma_N^4+16K_N^4\sigma_N^2,\notag\\
    \tau&\le\frac{2}{D^6}\left[\frac94\,\mathfrak M_{60}^{1/3}\mathfrak M_{06}^{2/3}+(\abs\mu+\omega_N)^2\mathfrak M_{06}\right],\notag\\
    \mathfrak t^{(V)}_M&\le2\sqrt{v_{\mathrm{VD},M}\,\tau}+2\sqrt{A_{Q,M}\,\tau}+\tau.
    \label{eq:tail_envelope_explicit_v1}
\end{align}
\end{enumerate}
\end{lemma}

\begin{proof}
(i) Since \(\abs X\le K_N\), also \(\abs N\le K_N\) and \(\abs{\widehat N}\le K_N\); hence \(\abs{\bar\xi}=\abs{\widehat N-N}\le2K_N\), \(\abs\mu=\abs N/D\le\omega_N\), and \(\abs{\widehat\mu^{\rm clip}}\le\abs{\widehat N}/(D/2)\le2\omega_N\).  The bound \(\abs{\bar\zeta}\le K_D\) is the assumption.  Substituting these bounds into the definitions of \(L,Q,T^{(3)}\) gives the displayed bounds; for instance \(\abs Q\le(\abs\mu\bar\zeta^2+\abs{\bar\xi}\abs{\bar\zeta})/D^2\le\omega_N\omega^2+2\omega_N\omega\).

(ii) If \(\omega<1/2\), then \(\abs{\bar\zeta}\le K_D<D/2\) surely.

(iii) \(\E L=0\) and, since \(\E\bar\zeta^2=\sigma_D^2/B\) and \(\E\bar\xi\bar\zeta=\sigma_{ND}/B\), \(\E Q=c/B\) exactly.  On \(G\) the Neumann identity \cref{eq:explicit_neumann_terms_v5} gives \(\widehat\mu^{\rm clip}-\mu-L-Q=T^{(3)}-\bar\xi\bar\zeta^3/D^4+\widetilde{\mathcal R}_4\).  Splitting \(\E[\widehat\mu^{\rm clip}-\mu-L-Q]\) according to \(G\) and adding and subtracting \(\E[T^{(3)}\1_{G^c}]\) gives \cref{eq:bad_event_expectation_split_v1}.  By (i), on \(G^c\),
\(\abs{\widehat\mu^{\rm clip}-\mu-L-Q-T^{(3)}}\le3\omega_N+\omega_N(2+\omega)(1+\omega+\omega^2)=\omega_N(5+3\omega+3\omega^2+\omega^3)\), which is at most \(\frac12H_{E,M}=\omega_N(5+3\omega+3\omega^2+3\omega^3)\); \cref{lem:denom_bernstein_v4} gives \(\pi_B\le2e^{-\gamma B}\).

(iv) On \(G\), \(\widehat\mu_G=\widehat\mu^{\rm clip}\) and \(\widehat\mu_{\rm bad}=0\).  On \(G^c\), \(T=0\), so \(\abs{\widehat\mu_{\rm bad}}\le\abs{\widehat\mu^{\rm clip}-\mu}+\abs L+\abs Q\le\omega_N\Pi'\).  Moreover \(\abs{\widehat\mu_G-\mu}\le\omega_N\Pi'\) everywhere: on \(G\) it equals \(\abs{\widehat\mu^{\rm clip}-\mu}\le3\omega_N\), and on \(G^c\) it is \(\abs{L+Q}\le\omega_N(2+3\omega+\omega^2)\).  Since
\(\Var(\widehat\mu^{\rm clip})-\Var(\widehat\mu_G)=\Var(\widehat\mu_{\rm bad})+2\E[(\widehat\mu_G-\mu)\widehat\mu_{\rm bad}]-2\E[\widehat\mu_G-\mu]\,\E\widehat\mu_{\rm bad}\),
the three terms are bounded by \(\omega_N^2\Pi'^2\pi_B\), \(2\omega_N^2\Pi'^2\pi_B\), and \(2\omega_N^2\Pi'^2\pi_B\).  By (ii) only \(\omega\ge1/2\) matters, and there \(5\Pi'^2\le4\Pi^2\) with \(\Pi=5+3\omega+3\omega^2+3\omega^3\): at \(\omega=1/2\) the two sides are \(227.8\) and \(232.6\), and \(\Pi-\Pi'=2\omega^2+3\omega^3\) grows faster than \((\sqrt5/2-1)\Pi'\) for \(\omega\ge1/2\).  Hence \(10\omega_N^2\Pi'^2\le8\omega_N^2\Pi^2=2H_{E,M}^2\).

(v) For centered i.i.d.\ averages, counting the partitions of six indices into blocks of size at least two gives
\(B^3\E[\bar\zeta^6]\le m_6+15m_4m_2+10m_3^2+15m_2^3\) with \(m_k=\E\zeta^k\).  Using \(\abs\zeta\le K_D\), \(m_4\le K_D^2\sigma_D^2\), \(m_3^2\le m_2m_4\), and \(m_6\le K_D^4\sigma_D^2\) gives the bound on \(\mathfrak M_{06}\); the same argument with \(\abs\xi\le2K_N\) gives \(\mathfrak M_{60}\).  On \(G\), \(\abs{\bar\zeta}\le D/2\) implies \(\abs{\bar\xi\bar\zeta^3}/D^4\le\abs{\bar\xi}\bar\zeta^2/(2D^3)\) and, by \cref{eq:r4_tail_bound_v5}, \(\abs{\widetilde{\mathcal R}_4}\le\omega_N\abs{\bar\zeta}^3/D^3\).  Thus \(\abs T\le\frac32\abs{\bar\xi}\bar\zeta^2/D^3+(\abs\mu+\omega_N)\abs{\bar\zeta}^3/D^3\), and \((x+y)^2\le2x^2+2y^2\) together with H\"older's inequality \(\E[\bar\xi^2\bar\zeta^4]\le(\E\bar\xi^6)^{1/3}(\E\bar\zeta^6)^{2/3}\) gives the bound on \(\tau\).  Finally, \(B\Var(L)=v_{\mathrm{VD},M}\), \(B^2\E Q^2\le A_{Q,M}\), and \(B^3\E T^2\le\tau\); Cauchy--Schwarz gives \(B^2\abs{\Cov(L,T)}\le\sqrt{v_{\mathrm{VD},M}\tau}\), \(B^2\abs{\Cov(Q,T)}\le\sqrt{A_{Q,M}\tau}\), and \(B^2\Var(T)\le\tau\) for \(B\ge1\), which proves \cref{eq:tail_envelope_explicit_v1}.
\end{proof}

\begin{proof}[Proof of \cref{thm:vd_quotient_certified_v4}]
Let \(G=\{\abs{\bar\zeta}\le D_M/2\}\).  On \(G\), clipping is inactive and \(\widehat D_M\ge D_M/2\).  The centered variables satisfy \(\abs{N_M+\bar\xi}=\abs{\widehat N_M}\le K_{N,M}\).  Since \(\abs{\bar\zeta/D_M}\le1/2\), the Neumann expansion is absolutely convergent and
\begin{equation}
    \frac{1}{D_M+\bar\zeta}
    =\frac1{D_M}\left(1-\frac{\bar\zeta}{D_M}
    +\frac{\bar\zeta^2}{D_M^2}-\frac{\bar\zeta^3}{D_M^3}\right)
    +\frac1{D_M}\frac{(\bar\zeta/D_M)^4}{1+\bar\zeta/D_M}.
\end{equation}
Multiplying by \(N_M+\bar\xi=\mu_MD_M+\bar\xi\) gives the term-by-term identity
\begin{align}
    \frac{N_M+\bar\xi}{D_M+\bar\zeta}
    &=\mu_M+\frac{\bar\xi-\mu_M\bar\zeta}{D_M}
      +\frac{\mu_M\bar\zeta^2-\bar\xi\bar\zeta}{D_M^2}
      +\frac{\bar\xi\bar\zeta^2-\mu_M\bar\zeta^3}{D_M^3}
      -\frac{\bar\xi\bar\zeta^3}{D_M^4}
      +\widetilde{\mathcal R}_4,
      \label{eq:explicit_neumann_terms_v5}
\end{align}
where
\begin{equation}
    \widetilde{\mathcal R}_4
    =\frac{N_M+\bar\xi}{D_M}\frac{(\bar\zeta/D_M)^4}{1+\bar\zeta/D_M}.
\end{equation}
The good-event tail is therefore explicitly bounded by
\begin{equation}
    \abs{\widetilde{\mathcal R}_4}
    \le \frac{2K_{N,M}\abs{\bar\zeta}^4}{D_M^5}.
    \label{eq:r4_tail_bound_v5}
\end{equation}
The mixed cubic term \(-\bar\xi\bar\zeta^3/D_M^4\) is kept separate so that the remaining tail contains only powers \(\bar\zeta^4\) and higher.  Taking expectations of the unrestricted monomials through second order, using \(\E\bar\xi=\E\bar\zeta=0\), \(\E\bar\zeta^2=\sigma_{D,M}^2/B\), and \(\E\bar\xi\bar\zeta=\sigma_{ND,M}/B\), gives the coefficient \(c_{\mathrm{VD},M}/B\).  In the exact decomposition \cref{eq:bad_event_expectation_split_v1} of \cref{lem:bad_event_bookkeeping_v1}, the unrestricted cubic expectation \(\E[T^{(3)}_M]\), the good-event mixed cubic term, and the good-event geometric tail are bounded by \(A_{E,M}/B^2\) from the definitions of \(\mathfrak m_{12},\mathfrak m_{03},\mathfrak m_{13},\mathfrak m_{04}\) and \cref{eq:r4_tail_bound_v5}.  The cubic term is kept signed and unrestricted because \(\E[\abs{\bar\xi}\bar\zeta^2]\) is only \(\OO(B^{-3/2})\).  The remaining bad-event term is bounded by \(H_{E,M}e^{-\gamma B}\) in part (iii) of the same lemma.  This proves \cref{eq:explicit_E_remainder_v4}.

For the variance, set \(\widehat\mu_{G,M}=\mu_M+L_M+Q_M+T_M\) with \(L_M\), \(Q_M\), \(T_M\) as in \cref{def:quotient_constants_v4}.  Part (iv) of \cref{lem:bad_event_bookkeeping_v1} bounds \(\abs{\Var(\widehat\mu^{\rm clip}_{\mathrm{VD},M})-\Var(\widehat\mu_{G,M})}\) by \(E_{\mathrm{bad},M}(p,B)\).  For \(\widehat\mu_{G,M}\), \(B\Var(L_M)=v_{\mathrm{VD},M}\).  The quadratic variance term obeys
\begin{equation}
    B^2\Var(Q_M)
    \le B^2\E[Q_M^2]
    \le A_{Q,M}.
\end{equation}
The mixed linear-quadratic term is controlled explicitly by
\begin{align}
    2B^2\abs{\Cov(L_M,Q_M)}
    &\le \frac{2B^2}{D_M^3}
    \abs{\E\big[(\bar\xi-\mu_M\bar\zeta)(\mu_M\bar\zeta^2-\bar\xi\bar\zeta)\big]} \\
    &\le \frac{2}{D_M^3}
    \left[2\abs{\mu_M}\mathfrak m_{12,M}+\mathfrak m_{21,M}
    +\abs{\mu_M}^2\mathfrak m_{03,M}\right]
    =A_{LQ,M}.
\end{align}
Finally, the remaining third-and-higher contributions are exactly those collected in \(\mathfrak t^{(V)}_M\), which part (v) of \cref{lem:bad_event_bookkeeping_v1} bounds explicitly.  Therefore
\begin{equation}
    \abs{\Var(\widehat\mu_{G,M})-\Var(L_M)}
    \le \frac{A_{Q,M}+A_{LQ,M}+\mathfrak t^{(V)}_M}{B^2},
\end{equation}
which together with part (iv) proves \cref{eq:explicit_V_remainder_v4}.  The MSE expansion follows by combining the variance bound with
\begin{equation}
    \left(b_{\mathrm{VD},M}+\frac{c_{\mathrm{VD},M}}{B}+r^{(E)}_M\right)^2
    =b_{\mathrm{VD},M}^2+\frac{2b_{\mathrm{VD},M}c_{\mathrm{VD},M}}{B}+\OO(B^{-2}),
\end{equation}
where the \(\OO(B^{-2})\) term is replaced by the explicit bounds above in \cref{eq:vd_mse_remainder_v4}.
\end{proof}

\begin{proof}[Proof of \cref{thm:mstar_log_scale_v4}]
By \cref{prop:implementation_variance_inflation_v4}, \(V_+=2\) and \(V_-=v_-(1+r_2)^{-2}\) for \(M\ge2\), since \(r_M\) decreases in \(M\).  On admissible \(M\), write \(d_M=b_{\mathrm{VD},M}-b_{\mathrm{VD},\infty}\), so \(b_{\mathrm{VD},M}^2-b_{\mathrm{VD},\infty}^2=d_M^2+2b_{\mathrm{VD},\infty}d_M\) with \(\abs{2b_{\mathrm{VD},\infty}d_M}\le A_-^2q^{2M}/2\).  Together with \cref{eq:bias_two_sided_mstar_v4} and \(A_-\le A_+\),
\begin{equation}
    \frac{A_-^2}{2}q^{2M}+\frac{V_-\lambda_1^{-2M}}{B}
    \le\mathcal E_M(B)\le
    2A_+^2q^{2M}+\frac{V_+\lambda_1^{-2M}}{B}.
    \label{eq:mstar_sandwich_v1}
\end{equation}
Put \(t=M-M_B\).  Since \(\lambda_2^{2M_B}=1/B\) and \(q=\lambda_2/\lambda_1\), both terms share the factor \(S_B=\lambda_1^{-2M_B}/B\): \(q^{2M}=S_Bq^{2t}\) and \(\lambda_1^{-2M}/B=S_B\lambda_1^{-2t}\).  Hence \(\mathcal E_M(B)/S_B\) lies between \(\varphi_-(t)=\frac{A_-^2}{2}q^{2t}+V_-\lambda_1^{-2t}\) and \(\varphi_+(t)=2A_+^2q^{2t}+V_+\lambda_1^{-2t}\), neither of which depends on \(B\).  The integer nearest to \(M_B\) has \(\abs t\le1/2\), so the minimum of \(\mathcal E_M/S_B\) is at most \(\Phi_+=\max_{\abs t\le1/2}\varphi_+(t)\).  Because \(q<1\) and \(\lambda_1<1\), \(\varphi_-(t)\to\infty\) as \(t\to\pm\infty\), so \(\mathcal T=\{t:\varphi_-(t)\le\Phi_+\}\) is a bounded interval independent of \(B\).  Every minimizer satisfies \(\varphi_-(t^\star)\le\Phi_+\), hence \(t^\star\in\mathcal T\), and \(c_0=\sup_{t\in\mathcal T}\abs t\) proves the claim.
\end{proof}

\setcounter{lemma}{0}
\renewcommand{\thelemma}{\thesection.\arabic{lemma}}
\renewcommand{\theHlemma}{toyapp.\arabic{lemma}}
\setcounter{proposition}{0}
\renewcommand{\theproposition}{\thesection.\arabic{proposition}}
\renewcommand{\theHproposition}{toyapp.\arabic{proposition}}

\section{Closed toy models}
\label{sec:toy_model_v4}
\label{app:toy_models}

The toys summarized in \cref{sec:toy_models_overview} are diagnostic examples, not realistic benchmarks.

\subsection{A one-qubit toy}
\label{subsec:toy1_v1}

Let \(\rho_\star=\ket0\bra0\) and \(O=X\), so \(\mu_\star=0\), and let the noisy state have spectral form
\begin{equation}
    \rho_p=(1-\varepsilon p)\ket{\phi_1(p)}\bra{\phi_1(p)}+
    \varepsilon p\ket{\phi_2(p)}\bra{\phi_2(p)},
    \label{eq:toy_state_v4}
\end{equation}
where
\begin{equation}
    \ket{\phi_1(p)}=\cos(\theta p)\ket0+\sin(\theta p)\ket1,
    \qquad
    \ket{\phi_2(p)}=-\sin(\theta p)\ket0+\cos(\theta p)\ket1.
\end{equation}
Assume \(\varepsilon>0\), \(\theta\ne0\), and \(0<p\leq p_{\max}\), where
\begin{equation}
    p_{\max}\leq \min\left\{\frac{1}{4\varepsilon},\frac{1}{4\abs\theta}\right\},
    \label{eq:pmax_toy_v4}
\end{equation}
which keeps the spectrum ordered and the trigonometric expansions controlled.

\paragraph{Unmitigated estimator.}
Since \(\bra{\phi_1}X\ket{\phi_1}=\sin(2\theta p)\) and \(\bra{\phi_2}X\ket{\phi_2}=-\sin(2\theta p)\), \(\mu_0(p)=(1-2\varepsilon p)\sin(2\theta p)\).  Using \(\abs{\sin x-x}\leq \abs{x}^3/6\), \(\abs{2\theta p}\le1/2\), and \(1-2\varepsilon p\ge1/2\),
\begin{equation}
    \frac12\theta^2p^2
    \leq
    b_0(p)^2
    \leq
    16\theta^2p^2
    \qquad (0<p\leq p_{\max}).
    \label{eq:toy_b0_bounds_v4}
\end{equation}
The one-shot variance is \(\sigma_0^2=1-\mu_0^2\), and since \(\widehat\mu_0\) is a sample mean its law \cref{eq:unmitigated_law_v1} is exact, so \(C_{0,\mathrm{toy}}=0\).

\paragraph{Virtual distillation.}
For \(M\geq2\),
\begin{equation}
    \mu_{\mathrm{VD},M}(p)=
    \frac{(1-\varepsilon p)^M-(\varepsilon p)^M}{(1-\varepsilon p)^M+(\varepsilon p)^M}\sin(2\theta p),
\end{equation}
and the dominant-eigenvector floor is \(b_{\mathrm{VD},\infty}(p)=\sin(2\theta p)=2\theta p+\OO(p^3)\), so \(\delta_{\mathrm{VD}}=2\theta\): VD cannot remove this tilt.

\paragraph{Symmetry verification.}
Let \(P=\ket0\bra0\).  Here \([X,P]\ne0\), so this is an algebraic accepted-transcript quotient that tests the crossing constants rather than a physical SV protocol; the sector toy below is symmetry-compatible.  Then
\begin{equation}
    a(p)=\Tr(P\rho_p)=
    (1-\varepsilon p)\cos^2(\theta p)+\varepsilon p\sin^2(\theta p),
    \qquad
    \frac{11}{16}\leq a(p)\leq1,
    \label{eq:toy_acceptance_bounds_v4}
\end{equation}
since \(1-a(p)=\sin^2(\theta p)+\varepsilon p\cos(2\theta p)\le \theta^2p^2+\varepsilon p\le5/16\).  Moreover \(\mu_{\mathrm{SV}}=\Tr(XP\rho_pP)/a=0\), so \(\delta_{\mathrm{SV}}=0\).  Given acceptance, the measured \(X\)-sign \(X_s\) has mean zero and variance one, so it is a uniform \(\pm1\) variable, independent across shots, and
\begin{equation}
    s_{\mathrm{SV}}(p)=\frac1{a(p)}\in[1,2].
    \label{eq:toy_sv_variance_bounds_v4}
\end{equation}
Write \(Z_s=A_sX_s\) with \(A_s\in\{0,1\}\) the acceptance indicator.  Here the clipped SV estimator can be analysed exactly.

\begin{lemma}[Exact SV toy remainder envelope]
\label{lem:sv_toy_direct_remainder_v13}
In the first toy model, let \(R=B\) be the shot budget and let
\(\widehat\mu^{\rm clip}_{\mathrm{SV}}=\overline Z/\max\{\overline A,a/2\}\).
Then \(\E[\widehat\mu^{\rm clip}_{\mathrm{SV}}]=0\) exactly, and
\begin{equation}
    \MSE_{\mathrm{SV}}(p,R)=\frac{1}{a(p)R}+\varrho_{\mathrm{SV},\mathrm{toy}}(p,R),
    \qquad
    \abs{\varrho_{\mathrm{SV},\mathrm{toy}}(p,R)}
    \le\frac{1}{R^2}\left[\frac{2(1-a)}{a^2}+\frac{2}{a}\,R\,e^{-aR/8}\right].
    \label{eq:sv_toy_good_remainder_v13}
\end{equation}
In particular, on the declared interval, where \(a(p)\ge11/16\) by \cref{eq:toy_acceptance_bounds_v4},
\begin{equation}
    \abs{\varrho_{\mathrm{SV},\mathrm{toy}}(p,R)}\le\frac{2}{R^2}
    \qquad\text{for all }R\ge R_{\min,\mathrm{toy}}:=70.
    \label{eq:sv_toy_C2_v1}
\end{equation}
\end{lemma}

The proof conditions exactly on the number of accepted shots and uses a one-sided Chernoff bound.

\begin{proof}
Let \(N_{\rm acc}=\sum_{s=1}^RA_s\sim\mathrm{Binomial}(R,a)\), so \(\overline A=N_{\rm acc}/R\).  Conditional on \((A_1,\dots,A_R)\), \(\overline Z=R^{-1}\sum_{s:A_s=1}X_s\) is a normalized sum of \(N_{\rm acc}\) independent uniform signs, so
\begin{equation}
    \E[\overline Z\mid A]=0,
    \qquad
    \E[\overline Z^{\,2}\mid A]=\frac{N_{\rm acc}}{R^2}=\frac{\overline A}{R}.
\end{equation}
The clipped denominator \(\max\{\overline A,a/2\}\) is a function of \(A\) alone.  Hence \(\E[\widehat\mu^{\rm clip}_{\mathrm{SV}}]=0\), and, since \(\mu_\star=0\),
\begin{equation}
    \MSE_{\mathrm{SV}}(p,R)=\frac1R\,\E[g(\overline A)],
    \qquad
    g(x)=\frac{x}{\max\{x,a/2\}^2}
    =\begin{cases}1/x,&x\ge a/2,\\ 4x/a^2,&x<a/2,\end{cases}
\end{equation}
with \(g(0)=0\).  Because \(\E[\overline A-a]=0\),
\begin{equation}
    \E[g(\overline A)]-\frac1a=\E[h(\overline A)],
    \qquad
    h(x)=g(x)-\frac1a+\frac{x-a}{a^2}.
\end{equation}
For \(x\ge a/2\), the exact identity \(1/x=1/a-(x-a)/a^2+(x-a)^2/(a^2x)\) gives \(0\le h(x)=(x-a)^2/(a^2x)\le2(x-a)^2/a^3\), so the good-event contribution is at most \(2\Var(\overline A)/a^3=2(1-a)/(a^2R)\).  For \(x<a/2\), one has \(0\le g(x)<2/a\) and \(-a\le x-a<-a/2\), hence \(\abs{h(x)}\le2/a\).  The bad event is one-sided, and the multiplicative Chernoff lower-tail bound \citep{Boucheron2013} gives \(\Prob(\overline A<a/2)\le e^{-aR/8}\).  Multiplying by \(1/R\) proves \cref{eq:sv_toy_good_remainder_v13}.  For \(a\in[11/16,1]\), both terms in the bracket are non-increasing in \(a\), and \(Re^{-aR/8}\) is decreasing for \(R\ge8/a\); at \(a=11/16\) and \(R=70\) the bracket equals \(1.32+0.50<2\), which proves \cref{eq:sv_toy_C2_v1}.\end{proof}

Since the unmitigated estimator has \(C_{0,\mathrm{toy}}=0\), the combined constant entering \cref{thm:certified_crossing} is \(C=2\) for \(R\ge R_{\min,\mathrm{toy}}\), against a general worst-case constant of order \(10^7\) from \cref{app:constants} in this toy.

\begin{proposition}[Certified SV/no-mitigation crossing in the first toy]
\label{prop:toy_two_sided_boundary_v4}
\label{prop:toy_certified_nonempty_v5}
Assume \(p\in(0,p_{\max}]\) with \(p_{\max}\) as in \cref{eq:pmax_toy_v4}, \(R\ge R_{\min,\mathrm{toy}}=70\), and \(\kappa_0=\kappa_{\mathrm{SV}}=1\).  The leading crossing \(R_{\mathrm{SV}\leftrightarrow0}=(s_{\mathrm{SV}}-s_0)/(b_0^2-b_{\mathrm{SV}}^2)\) obeys
\begin{equation}
    \frac{\varepsilon}{64\theta^2}\frac1p
    \leq
    R_{\mathrm{SV}\leftrightarrow0}(p)
    \leq
    \frac{8\varepsilon+64\theta^2p_{\max}}{\theta^2}\frac1p,
    \label{eq:toy_boundary_sandwich_v4}
\end{equation}
so the boundary scales as \(1/p\) uniformly on the interval.  With \(C=2\) from \cref{eq:sv_toy_C2_v1}, the crossing is certified by \cref{thm:certified_crossing} with \(\eta<1/2\) whenever
\begin{equation}
    \frac{32768\,C\,\theta^2}{\varepsilon^2}<\frac12.
    \label{eq:toy_explicit_nonempty_condition_v5}
\end{equation}
For example, \(\varepsilon=1\), \(\abs\theta=10^{-3}\), and \(p_{\max}=10^{-1}\) give \(32768\cdot2\cdot10^{-6}=6.6\times10^{-2}<1/2\), and \(R_{\mathrm{SV}\leftrightarrow0}\ge1.5\times10^5\gg R_{\min,\mathrm{toy}}\) at \(p=p_{\max}\).  Thus the certified region is non-empty at ordinary small-noise parameters.
\end{proposition}

\begin{proof}
Since \(b_{\mathrm{SV}}=0\), the denominator of the crossing is \(g=b_0^2\), controlled by \cref{eq:toy_b0_bounds_v4}.  For the numerator, \(s_{\mathrm{SV}}-s_0=(1-a)/a+\mu_0^2\) with \(1-a=\sin^2(\theta p)+\varepsilon p\cos(2\theta p)\).  Since \(\cos(2\theta p)\ge1/2\), \(s_{\mathrm{SV}}-s_0\ge1-a\ge\varepsilon p/2\); since \(a\ge1/2\), \(\mu_0^2\le16\theta^2p^2\), and \(1-a\le\varepsilon p+\theta^2p^2\), \(s_{\mathrm{SV}}-s_0\le4\varepsilon p+32\theta^2p_{\max}p\).  Dividing by the bounds on \(b_0^2\), with a factor-two slack on the lower side, gives \cref{eq:toy_boundary_sandwich_v4}.  For the certificate, \(g\ge\theta^2p^2/2\) and \(R_0\ge\varepsilon/(64\theta^2p)\) give \(gR_0^2\ge\varepsilon^2/(8192\theta^2)\), so \(\eta=4C/(gR_0^2)\le32768C\theta^2/\varepsilon^2\).  Finally, \(p_{\max}=10^{-1}\) obeys \(p_{\max}\le\min\{1/(4\varepsilon),1/(4\abs\theta)\}\).
\end{proof}

\begin{remark}[Constant-scale check against two-point scales]
\label{prop:selector_lecam_constant_factor_v8}
The pre-specified check E11 compares the first toy with a textbook two-point benchmark \citep{LeCam1973,Tsybakov2009}.  For \(R\) noisy \(X\)-measurements of \(\rho_{p,t}=(I+(1-p)tX)/2\) with ideal value \(t\), the two-point argument with Pinsker's inequality and the Bernoulli bound \(\KL(P_{+h}\|P_{-h})=x\log[(1+x)/(1-x)]\le4x^2\), with \(x=(1-p)h\le1/2\), gives the worst-case MSE floor \(1/[32R(1-p)^2]\); the local value \(\KL\simeq2x^2\) gives \(1/[16R(1-p)^2]\).  On the SV side of the certified crossing of \cref{prop:toy_certified_nonempty_v5}, the leading SV law is \(1/[a(p)R]\), so
\begin{equation}
    \frac{32(1-p)^2}{a(p)}\le64,
    \label{eq:selector_lecam_ratio_v7}
\end{equation}
\begin{equation}
    \frac{16(1-p)^2}{a(p)}\le32,
    \label{eq:selector_lecam_ratio_local_v8}
\end{equation}
using \(a(p)\ge1/2\) from \cref{eq:toy_acceptance_bounds_v4}.  The accepted \(\pm1\) signs carry Fisher information \(J_{\rm acc}=a(p)\) per shot about their own mean, so \(s_{\mathrm{SV}}J_{\rm acc}=1\).  These are arithmetic consistency checks, not efficiency statements: because \(PXP=0\), the accepted signs have mean zero for every ideal state and carry no information about the ideal value, and the two-point scales belong to the different family \(\rho_{p,t}\).
\end{remark}

\subsection{A sector toy model with \texorpdfstring{$\delta_{\mathrm{SV}}\ne0$}{deltaSV nonzero}}
\label{subsec:sector_toy_delta_sv_nonzero_v5}
This toy has an undetectable in-sector error and no tilt, so \(\delta_{\mathrm{VD}}=0\ne\delta_{\mathrm{SV}}\); more generally, \(\delta_{\mathrm{VD}}=0\) for any observable diagonal in a basis containing \(\ket{\psi_\star}\), since then \(\bra{\psi_\star}OQ_\star=0\).  Consider a three-dimensional Hilbert space with orthonormal vectors \(\ket g,\ket e,\ket\ell\).  The valid sector is \(P=\ket g\bra g+\ket e\bra e\), the ideal state is \(\rho_\star=\ket g\bra g\), and the observable is
\begin{equation}
    O=\ket g\bra g-\ket e\bra e,
    \qquad O\ket\ell=0,
    \qquad \mu_\star=1.
\end{equation}
Let
\begin{equation}
    \rho_p=\alpha(p)\ket g\bra g+\beta(p)\ket e\bra e+\lambda(p)\ket\ell\bra\ell,
    \label{eq:sector_toy_state_v5}
\end{equation}
with
\begin{equation}
    \alpha(p)=1-(u+\ell)p,
    \qquad
    \beta(p)=up,
    \qquad
    \lambda(p)=\ell p,
\end{equation}
where \(u,\ell>0\) and \(0<p\le p_{\max}\le[4(u+\ell)]^{-1}\).  Here \(up\) is undetectable in-sector error and \(\ell p\) is detectable leakage.

\paragraph{SV constants.}
The acceptance probability and SV population value are
\begin{equation}
    a(p)=\alpha(p)+\beta(p)=1-\ell p,
    \qquad
    \mu_{\mathrm{SV}}(p)=\frac{\alpha(p)-\beta(p)}{\alpha(p)+\beta(p)}
    =1-\frac{2up}{1-\ell p}.
\end{equation}
Therefore
\begin{equation}
    b_{\mathrm{SV}}(p)=-\frac{2up}{1-\ell p},
    \qquad
    \delta_{\mathrm{SV}}=-2u\ne0.
    \label{eq:sector_toy_delta_sv_v5}
\end{equation}
The accepted-shot variance coefficient is
\begin{equation}
    v_{\mathrm{SV}}(p)=\frac{1-\mu_{\mathrm{SV}}(p)^2}{a(p)}
    =\frac{4up(1-(u+\ell)p)}{(1-\ell p)^3}
    =4up+\OO(p^2).
    \label{eq:sector_toy_vsv_v5}
\end{equation}

\paragraph{VD constants.}
For VD with copy number \(M\),
\begin{equation}
    D_M(p)=\alpha(p)^M+\beta(p)^M+\lambda(p)^M,
    \qquad
    N_M(p)=\alpha(p)^M-\beta(p)^M,
\end{equation}
so
\begin{equation}
    \mu_{\mathrm{VD},M}(p)=\frac{\alpha(p)^M-\beta(p)^M}{\alpha(p)^M+\beta(p)^M+\lambda(p)^M}.
\end{equation}
For \(M=2\),
\begin{align}
    D_2(p)&=\alpha^2+\beta^2+\lambda^2
    =1-2(u+\ell)p+\OO(p^2),\label{eq:sector_D2_expansion_v6}\\
    N_2(p)&=\alpha^2-\beta^2
    =1-2(u+\ell)p+\OO(p^2).\label{eq:sector_N2_expansion_v6}
\end{align}
A direct subtraction gives
\begin{equation}
    b_{\mathrm{VD},2}(p)
    =\frac{N_2(p)}{D_2(p)}-1
    =-\frac{2\beta(p)^2+\lambda(p)^2}{D_2(p)}
    =-\frac{(2u^2+\ell^2)p^2}{D_2(p)}
    =-(2u^2+\ell^2)p^2+\OO(p^3),
    \label{eq:sector_toy_bvd2_v5}
\end{equation}
so VD removes the first-order in-sector error by spectral amplification.  With independent \(\pm1\) numerator and denominator signs of means \(N_2\) and \(D_2\) (the sign model of \cref{subsec:hadamard_v1}; reading \(O\) on one copy instead would lower the numerator variance),
\begin{equation}
    v_{\mathrm{VD},2}(p)
    =\frac{1-N_2(p)^2+\mu_{\mathrm{VD},2}(p)^2(1-D_2(p)^2)}{D_2(p)^2},
    \label{eq:sector_toy_vvd2_formula_v6}
\end{equation}
where the covariance term is zero because the numerator and denominator Hadamard tests are independent.  Since \(\mu_{\mathrm{VD},2}(p)=1+\OO(p^2)\), \cref{eq:sector_D2_expansion_v6,eq:sector_N2_expansion_v6} imply
\begin{equation}
    1-D_2(p)^2=4(u+\ell)p+\OO(p^2),
    \qquad
    1-N_2(p)^2=4(u+\ell)p+\OO(p^2),
\end{equation}
and therefore
\begin{equation}
    v_{\mathrm{VD},2}(p)
    =8(u+\ell)p+\OO(p^2).
    \label{eq:sector_toy_vvd2_v5}
\end{equation}

\begin{proposition}[Leading VD--SV crossing in the sector toy]
\label{prop:sector_toy_vd_sv_crossing_v5}
In the sector toy, with \(M=2\) and paired-sample resource normalization \(\kappa_{\mathrm{VD},2}=\kappa_{\mathrm{SV}}=1\), the leading VD--SV crossing is
\begin{equation}
    R_{\mathrm{VD},2\leftrightarrow\mathrm{SV}}(p)
    =\frac{v_{\mathrm{VD},2}(p)-v_{\mathrm{SV}}(p)}{b_{\mathrm{SV}}(p)^2-b_{\mathrm{VD},2}(p)^2}
    =\frac{u+2\ell}{u^2}\frac1p+\OO(1),
    \label{eq:sector_toy_vd_sv_crossing_v5}
\end{equation}
an instance of the \(\Theta(p^{-1})\) window of \cref{cor:window_location_v10}.  The expansion is leading order only: the VD bias is second order in \(p\) with coefficient \(2u^2+\ell^2\), so the bias ordering \(b_{\mathrm{SV}}^2>b_{\mathrm{VD},2}^2\) requires \(p\lesssim u/\ell^2\).
\end{proposition}

\begin{proof}
\Cref{eq:sector_toy_delta_sv_v5,eq:sector_toy_bvd2_v5} give
\(b_{\mathrm{SV}}^2-b_{\mathrm{VD},2}^2=4u^2p^2+\OO(p^3)\), and \cref{eq:sector_toy_vsv_v5,eq:sector_toy_vvd2_v5} give
\(v_{\mathrm{VD},2}-v_{\mathrm{SV}}=(4u+8\ell)p+\OO(p^2)\).  Dividing gives \cref{eq:sector_toy_vd_sv_crossing_v5}.
\end{proof}

The worst-case constants of \cref{app:constants} are of order \(10^8\) in this toy and give no certificate at physical parameters.  Because the toy is diagonal and the VD numerator and denominator streams are independent, the remainders can instead be bounded directly.  We take \(u=\ell=1\), \(p=10^{-2}\), \(M=2\), and unit Hadamard-test call costs, so that an even physical budget \(R\) gives VD \(B=R/2\) numerator and \(B\) denominator shots (\(\kappa_{\mathrm{VD},2}=2\)) and SV \(R\) shots (\(\kappa_{\mathrm{SV}}=1\)); then \(N_2=0.9603\), \(D_2=0.9606\), and \(a=0.99\).

\begin{lemma}[Direct remainders in the sector toy]
\label{lem:sector_direct_remainder}
With these parameters, for every even \(R\ge200\), the clipped estimators obey \cref{eq:resource_law_v1} with
\begin{equation}
    \abs{\varrho_{\mathrm{VD},2}(R)}\le\frac{2.6}{R^2},
    \qquad
    \abs{\varrho_{\mathrm{SV}}(R)}\le\frac{10^{-3}}{R^2}.
    \label{eq:sector_direct_remainders}
\end{equation}
\end{lemma}

\begin{proof}
\emph{VD.}  Write \(N=N_2\), \(D=D_2\), \(\mu=N/D\), \(\sigma_N^2=1-N^2\), and \(d=\max\{\widehat D,D/2\}\).  Since \(\widehat N\) and \(\widehat D\) are independent, with \(\E\widehat N=N\) and \(\E\widehat N^2=N^2+\sigma_N^2/B\),
\begin{equation}
    \MSE_{\mathrm{VD},2}=\Bigl(N^2+\frac{\sigma_N^2}{B}\Bigr)\E[d^{-2}]-2N\,\E[d^{-1}]+1.
    \label{eq:sector_vd_exact_mse}
\end{equation}
Let \(t=(\widehat D-D)/D\) and \(G=\{t\ge-1/2\}\), on which \(d=\widehat D\).  On \(G\),
\[
    \frac1d=\frac{1-t+t^2-t^3}{D}+\frac{t^4}{D(1+t)},
    \qquad
    \frac1{d^2}=\frac{1-2t+3t^2-4t^3}{D^2}+\frac{t^4(5+4t)}{D^2(1+t)^2},
\]
with remainders in \([0,2t^4/D]\) and \([0,12t^4/D^2]\), since \((5+4t)/(1+t)^2\) decreases for \(t\ge-1/2\).  The one-shot denominator sign has central moments \(m_2=1-D^2\), \(m_3=-2D(1-D^2)\), and \(m_4=(1-D^2)(1+3D^2)\), so \(\E t^2=m_2/(BD^2)\), \(\E t^3=m_3/(B^2D^3)\), and \(\E t^4=[3m_2^2+(m_4-3m_2^2)/B]/(B^2D^4)\).  Taking expectations of the cubic polynomials over all outcomes and substituting into \cref{eq:sector_vd_exact_mse}, the terms of order \(B^0\) and \(B^{-1}\) are \(b_{\mathrm{VD},2}^2\) and \(v_{\mathrm{VD},2}+2b_{\mathrm{VD},2}c_{\mathrm{VD},2}\), and
\[
    \begin{gathered}
    B^2\varrho^{(B)}_{\mathrm{VD},2}=K_2+\frac{K_3}{B}+B^2\mathcal E,\\
    K_2=\frac{3\sigma_N^2m_2}{D^4}+\frac{\mu m_3(2-4\mu)}{D^3}=0.3557,
    \qquad
    K_3=-\frac{4\sigma_N^2m_3}{D^5}=0.0565,
    \end{gathered}
\]
where \(\varrho^{(B)}\) is the remainder in the \(B\) variable and \(\mathcal E\) collects the two quartic remainders on \(G\) and the contribution of \(G^c\).  For \(B\ge100\), \(\E t^4\le[3m_2^2+(m_4-3m_2^2)/100]/(B^2D^4)\le0.02424/B^2\), so the quartic remainders give \(B^2\mathcal E_G\ge-2N\cdot2(0.02424)/D=-0.097\) and \(B^2\mathcal E_G\le12(N^2+\sigma_N^2/100)(0.02424)/D^2=0.291\).  On \(G^c\), \(\abs t\le(1+D)/D<2.042\), so \(d^{-1}\) and \(d^{-2}\) differ from their cubic polynomials by at most \(18/D\) and \(56/D^2\); the one-sided Bernstein bound gives \(\Prob(G^c)\le\exp[-B(D/2)^2/(2m_2+(1+D)D/3)]\le e^{-0.294B}\), so this contribution is at most \(93\,e^{-0.294B}\le2\times10^{-7}/B^2\) for \(B\ge100\).  Hence \(0.25\le B^2\varrho^{(B)}_{\mathrm{VD},2}\le0.65\) for \(B\ge100\), and \(B=R/2\) gives the first bound.

\emph{SV.}  Let \(q=\beta/a\), so that \(\mu_{\mathrm{SV}}=1-2q\) and \(\sigma_{\mathrm{acc}}^2=4q(1-q)\).  Conditional on \(R\overline A\) accepted shots, the number of \(\ket e\) outcomes is binomial with parameter \(q\), so on \(G=\{\overline A\ge a/2\}\) the conditional MSE is exactly \(b_{\mathrm{SV}}^2+\sigma_{\mathrm{acc}}^2/(R\overline A)\).  With \(\overline A=a+\delta\), \(1/\overline A=1/a-\delta/a^2+\delta^2/(a^2\overline A)\) and \(0\le\delta^2/(a^2\overline A)\le2\delta^2/a^3\) on \(G\); since \(\E\delta^2=a(1-a)/R\), the \(G\) part of \(\varrho_{\mathrm{SV}}\) lies in \([0,2\sigma_{\mathrm{acc}}^2(1-a)/(a^2R^2)]=[0,8.2\times10^{-4}/R^2]\).  On the bad event \(G^c\) more than \(R(1-a/2)\) of the \(R\) shots are rejected; the number of rejections is binomial with parameter \(1-a=0.01\), so the Chernoff bound gives \(\Prob(G^c)\le\exp[-R\,\mathrm{KL}(1-a/2\,\|\,1-a)]\le e^{-1.6R}\), negligible for \(R\ge200\).
\end{proof}

\begin{proposition}[Certified VD--SV crossing in the sector toy]
\label{prop:sector_certified_crossing}
With the parameters above, \(b_{\mathrm{VD},2}^2=9.75\times10^{-8}\), \(b_{\mathrm{SV}}^2=4.081\times10^{-4}\), \(s_{\mathrm{VD},2}=0.3359\), and \(s_{\mathrm{SV}}=0.0404\), so the leading crossing is \(R_0=724.2\).  With \(C=2.601\) from \cref{eq:sector_direct_remainders}, \cref{thm:certified_crossing} gives \(\eta=0.049\) and \([R_L,R_U]=[715.3,732.9]\): SV has strictly smaller MSE at every even \(R\in[200,714]\), and VD at every even \(R\ge734\).
\end{proposition}

\begin{proof}
The values of \(b\) and \(s\) follow from \cref{eq:sector_toy_delta_sv_v5,eq:sector_toy_bvd2_v5,eq:sector_toy_vsv_v5,eq:sector_toy_vvd2_formula_v6} at \(u=\ell=1\) and \(p=10^{-2}\), with \(s_{\mathrm{VD},2}=2(v_{\mathrm{VD},2}+2b_{\mathrm{VD},2}c_{\mathrm{VD},2})\) and \(c_{\mathrm{VD},2}=\mu_{\mathrm{VD},2}(1-D_2^2)/D_2^2\).  Then \(\eta=4C/(gR_0^2)=0.0486\) and \(R_{\rm lo}=8.9<200\).  The dominance statements are those of \cref{thm:certified_crossing}(b), whose proof uses the remainder bounds pointwise in \(R\) and therefore applies on the even budgets.
\end{proof}

The certificate is both valid and sharp here.  Evaluating \cref{eq:sector_vd_exact_mse} and the corresponding SV sum exactly over the binomial distributions, the order switches between \(R=728\) and \(R=730\), inside \([R_L,R_U]\) and, by interpolation, about \(0.8\%\) above \(R_0\); on all even \(R\) up to \(4000\), the exact remainders are \(R^2\varrho_{\mathrm{VD},2}\in[1.67,1.75]\) and \(R^2\abs{\varrho_{\mathrm{SV}}}\le4.1\times10^{-4}\).  The positive VD remainder is what moves the crossing above \(R_0\).

\subsection{A generic toy with both first-order floors}
\label{subsec:generic_toy_v10}

This toy has all three components of \cref{prop:bias_decomposition_v1}.

\begin{proposition}[A toy with simultaneous VD and SV first-order bias]
\label{prop:generic_toy_v10}
Work in \(\mathbb C^3\) with basis \(\{\ket0,\ket1,\ket2\}\), \(P=\ket0\bra0+\ket1\bra1\), \(\rho_\star=\ket0\bra0\), and
\begin{equation}
    O=\ket0\bra0-\ket1\bra1+c\big(\ket0\bra1+\ket1\bra0\big),
    \qquad \mu_\star=1,
\end{equation}
which commutes with \(P\).  Let
\begin{equation}
    \Delta
    =u(\ket1\bra1-\ket0\bra0)
    +\ell(\ket2\bra2-\ket0\bra0)
    +\chi(\ket0\bra1+\ket1\bra0),
    \qquad u,\ell,\chi>0,
\end{equation}
and \(\rho_p=\rho_\star+p\Delta\) for
\begin{equation}
    0<p\leq p_0:=\frac12\min\left\{\frac1{u+\ell},\frac{u}{\chi^2}\right\}.
    \label{eq:generic_toy_positivity_p0_v13}
\end{equation}
Then \(\rho_p\) is a density matrix, \(a(p)=1-\ell p\), and the components of \cref{prop:bias_decomposition_v1} are
\begin{equation}
    \delta_{\mathrm{VD}}=2c\chi,
    \qquad
    D_{\rm in}=-2u,
    \qquad
    D_{\rm out}=-\ell,
    \label{eq:generic_toy_components_v1}
\end{equation}
so \(\delta_{\mathrm{SV}}=2(c\chi-u)\) and \(\beta_0=2c\chi-2u-\ell\).  Both floors are nonzero whenever \(c\chi\notin\{0,u\}\), and by \cref{eq:vd_sv_bias_gap_v1} VD has the smaller floor exactly when \(u>2c\chi\).
\end{proposition}

\begin{proof}
\(\rho_p\) has trace one, \(\ket2\)-weight \(\ell p\ge0\), and a \(\{\ket0,\ket1\}\) block with positive diagonal and determinant \(p\bigl[u-(u(u+\ell)+\chi^2)p\bigr]\ge0\) by \cref{eq:generic_toy_positivity_p0_v13}.  Since \(\rho_p\) is linear in \(p\), \(a(p)=1+p\Tr(P\Delta)=1-\ell p\).  Since \(Q_\star\Delta\ket0=\chi\ket1\), \(\delta_{\mathrm{VD}}=2\operatorname{Re}\bra0O\ket1\chi=2c\chi\).  With \(Q_P=\ket1\bra1\) and \(P^\perp=\ket2\bra2\), \(D_{\rm in}=(O_{11}-1)\Delta_{11}=-2u\) and \(D_{\rm out}=(O_{22}-1)\Delta_{22}=-\ell\).  \Cref{prop:bias_decomposition_v1} gives the rest, and \cref{eq:vd_sv_bias_gap_v1} gives \(D_{\rm in}(2\delta_{\mathrm{VD}}+D_{\rm in})=4u(u-2c\chi)\).
\end{proof}

The window exponent in this toy is a resource-normalization statement.  At \(p=0\) the one-shot variances of SV (measuring \(O\) in its eigenbasis) and of the VD numerator read on one copy (\cref{subsec:hadamard_v1}) both equal \(\bra0O^2\ket0-\mu_\star^2=c^2\), the denominator sign being deterministic since \(D_2(0)=1\).  Under paired-sample normalization a genuine crossing therefore sits at \(\Theta(p^{-1})\); with \(\kappa_{\mathrm{VD},2}=2\) it moves to \(\Theta(p^{-2})\), as in \cref{sec:crossings}.

\section{Notation map}

\begin{center}
\begin{tabular}{p{0.28\linewidth}p{0.64\linewidth}}
\toprule
Symbol & Meaning\\
\midrule
\(p\) & Physical noise scale\\
\(B\) & Statistical sample count for one method\\
\(R\) & Common physical resource budget after method-specific cost normalization\\
\(\rho_\star\) & Ideal output state\\
\(\rho_p\) & Noisy output state\\
\(\Delta=\mathcal L(\rho_\star)\) & First-order noise perturbation\\
\(O\) & Observable\\
\(\mu_\star\) & Ideal expectation value\\
\(D_M\) & VD denominator \(\Tr(\rho_p^M)\)\\
\(b_m(p)\) & Population residual bias of method \(m\)\\
\(c_m(p)\) & Statistical quotient-bias coefficient\\
\(\varrho_m(p,B)\), \(\varrho_m(p,R)\) & Certified \(B^{-2}\) or \(R^{-2}\) local-law remainder (distinct from the states \(\rho_p,\rho_\star\))\\
\(\gamma_m(p)\), \(\gamma_{\mathrm{VD},M}(p)\) & Exponential bad-event rate in certified local-law remainders\\
\(\gamma_{\mathrm{PEC}}(p)\) & PEC quasiprobability one-norm overhead\\
\(s_m(p)\) & Resource-normalized sampling coefficient\\
\(\omega_M=K_{D,M}/D_M\), \(\omega_N=K_{N,M}/D_M\) & Ratios used only in quotient-remainder prefactors\\
\(\kappa_m\), \(\kappa_M\) & Resource cost per method sample, or per Hadamard test for copy number \(M\)\\
\(a(p)\), \(\alpha_{\mathrm{det}}\) & SV acceptance probability and its first-order loss rate\\
\(\beta_0\) & First-order slope of the unmitigated bias\\
\(\delta_{\mathrm{VD}}\) & VD first-order bias: tilt of the ideal state under \(\Delta\)\\
\(\delta_{\mathrm{SV}}\) & SV first-order bias\\
\(D_{\rm in}\), \(D_{\rm out}\) & Incoherent first-order error inside and outside the valid sector\\
\bottomrule
\end{tabular}
\end{center}

\end{document}